\documentclass{amsart}

\usepackage{enumerate}
\usepackage{amsmath,amssymb,amsfonts,amscd,amsthm,amsopn,epsfig}  
\usepackage{amsthm}
\usepackage{amssymb}
\usepackage{latexsym}
\usepackage{cite}

\usepackage{hyperref}
\hypersetup{citecolor=blue}

\newtheorem{defn}{Definition}[section]
\newtheorem{prop}{Proposition}[section]

\newtheorem{thrm}{Theorem}[section]
\theoremstyle{definition}
\newtheorem{remark}{\it Remark}[section]
\newtheorem{crl}{Corollary}[section]

\numberwithin{equation}{section}

\newcommand{\el}{\nonumber\\}

\def\sl{\mathfrak{sl}}
\def\gl{\mathfrak{gl}}

\def\g{\mathfrak{g}}
\def\h{\mathfrak{h}}
\def\m{\mathfrak{m}}

\def\bh{{\bf h}}
\def\bn{{\bf n}}

\def\ad{{\rm ad}}
\def\adr{{\rm ad}_r}

\def\cT{\mathcal{T}}
\def\cB{\mathcal{B}}
\def\cV{\mathcal{V}}
\def\cW{\mathcal{W}}
\def\cM{\mathcal{M}}

\def\cR{\mathcal{R}}
\def\cA{\mathcal{A}}
\def\cL{\mathcal{L}}
\def\cU{\mathcal{U}}

\def\tc{\otimes}
\def\tp{\oplus}

\def\Yg{\mathcal{Y}(\g)}
\def\Ygh{\mathcal{Y}(\g,\theta(\g))}
\def\Ygg{\mathcal{Y}(\g,\g)}

\def\glone{\cU_q(\widehat{\mathfrak{gl}}(1|1))}
\def\sltwo{\cU_q(\widehat{\mathfrak{sl}}(2))}

\def\Uq{\cU_{q}({\mathfrak{g}})}
\def\Ua{\cU_{q}(\widehat{\mathfrak{g}})}

\def\yone{\mathcal{Y}({\mathfrak{gl}}(1|1))}
\def\ytwo{\mathcal{Y}({\mathfrak{sl}}(2))}

\def\hp{h^+_2}

\def\Y{\mathcal{Y}}

\def\Vl{\cV_l}
\def\tth{\tilde\theta}
\def\bth{\bar\theta}

\def\sk{\quad\!}

\newcommand{\wt}[1]{\widetilde{#1}}
\newcommand{\wh}[1]{\widehat{#1}}
\newcommand{\dwh}[1]{\widehat{\widehat{#1}}}
\newcommand{\dwt}[1]{\widetilde{\widetilde{#1}}}
\newcommand{\bb}[1]{\mathbb{#1}}

\begin{document}

\title{Reflection algebras for SL(2) and GL(1$|$1)}

\author{Vidas Regelskis}

\email{vr509@york.ac.uk}

\address{Department of Mathematics University of York, York, YO10 5DD, U.K. and
Institute of Theoretical Physics and Astronomy of Vilnius University, Go\v{s}tauto 12, Vilnius 01108, Lithuania}


\begin{abstract}
We present a generalization the G.\ Letzter's theory of quantum symmetric pairs of semisimple Lie algebras for the case of quantum affine algebras. We then study solutions of the reflection equation for the quantum affine algebras $\sltwo$ and $\glone$ and their Yangian limit for singlet (diagonal) and vector (non-diagonal) boundary conditions. We construct the corresponding quantum affine coideal subalgebras that are based on the quantum symmetric pairs, and the (generalized) twisted Yangians.
\end{abstract}

\maketitle



\section*{Introduction}

Quantum affine algebras and Yangians are the simplest examples of the infinite-dimensional quantum groups and play a central role in quantum integrable systems \cite{FST,KRS}. These algebras were introduced in \cite{Dr,Jimbo} and since then had a significant impact on the development of the quantum inverse scattering method and related quantum integrable models.

The quantum affine algebra $\cU_q(\wh{\g})$ and the Yangian $\Yg$ are deformations of the universal enveloping algebras $\cU(\wh{\g})$ and $\cU(\g[u])$ respectively, where $\wh{\g}$ is the affine Kac--Moody algebra, a central extension of the Lie algebra of maps $\mathbb{C}^\times\to\g$ of a finite Lie algebra $\g$, and $\g[u]$ is a deformation of the Lie algebra of maps $\mathbb{C}\to\g$. In such a way, Yangians may be viewed as a specific degenerate limit of quantum affine algebras \cite{Dr87}. We refer to \cite{CP} for complete details on quantum groups.

In this paper we will concentrate on the finite dimensional representations of quantum groups, the so-called \textit{evaluation} representations. These are constructed via the epimorphisms $\text{ev}_a : \cU_q(\wh{\g})\to\cU_q(\g)$ and $\text{ev}_a : \Yg\to\cU(\g)$ called the evaluation homomorphisms, which evaluate the $\g$-valued polynomials at a point $a\in\mathbb{C}$. Such representations have wide applications in both mathematics and physics. However they can be constructed for some Lie algebras only.

We will consider integrable models with open boundary conditions, and concentrate on the algebras that define solutions of the reflection equation \cite{Sk}. These algebras are one-sided coideal subalgebras and are conveniently called reflection algebras. Such algebras have been extensively studied in e.g.\ \cite{Cher,Ols,MNO,MR,DMS,MacKay:2002at,Mu,Doikou:2004hy} and more recently in e.g.\ \cite{IT,BB,BC,GM}. For field theoretical applications of the coideal subalgebras and corresponding reflection matrices we refer to \cite{GZ,MacKay:2004tc}.

For this purpose we have selected two simple Lie (super-) algebras, $\sl(2)$ and $\gl(1|1)$, and have studied reflection algebras for the corresponding quantum groups: quantum affine (super-) algebras $\sltwo$ and $\glone$, and Yangians $\ytwo$ and $\yone$. We refer to \cite{CP} and \cite{Mo,JinfangCai,Zhang} for details on these algebras.
For each quantum group we consider \textit{singlet} and \textit{vector} boundary conditions. The singlet boundary forms singlet (trivial) representation of the boundary (reflection) algebra, thus there are no boundary degrees of freedom in the associated field theory. The vector boundary forms an evaluation representation of the boundary algebra and has boundary degrees of freedom. In field theoretical models the evaluation parameter is usually associated with the rapidity of a state. Boundary states have zero rapidity and thus the boundary evaluation parameter is conveniently associated with some conserved quantity of the boundary, e.g.\ the energy \cite{MReg}.
Singlet boundaries have been heavily studied and the corresponding boundary algebras and solutions of the reflection equation for most of the semisimple Lie (super-) algebras are well known. However vector boundaries have not been studied as much as the singlet ones. The corresponding reflection matrices are usually constructed via the boundary bootstrap procedure by fusing bulk R-matrices with an appropriate scalar reflection matrix (see e.g.\ \cite{Sk,Mezincescu:1991ke}), thus we hope the boundary algebras we have constructed in this paper will serve not only as neat examples but will contribute to the exploration of such boundaries as well.

We show that for the quantum affine algebras, the reflection algebras for both boundary conditions are quantum affine coideal subalgebras that are a generalization of the quantum symmetric pairs of simple Lie algebras considered in \cite{Le} and \cite{Le1}, where quantum symmetric pairs and the associated coideal subalgebras for all simple Lie algebras have been classified. Such coideal subalgebras were also considered in \cite{Kolb,KS,HK,LMR}. Similar quantum affine coideal subalgebras for algebra of type $A^{(1)}_1$ was considered in \cite{DN,IT}, for $D_n^{(1)}$ in \cite{DeliusGeorge}, for the double affine Hecke algebras of type $C^\vee C_n$ in \cite{JM}, for $\mathfrak{o}(n)$ and $\mathfrak{sp}(2n)$ in \cite{MRS}, and for the Sine-Gordon and affine Toda field theories in \cite{Delius:1998rf,Delius:2001qh,Baseilhac:2002kf,BB}.

For the Yangian case we will consider two types of the (generalized) twisted Yangians.
For the singlet boundary we will consider the twisted Yangian introduced in \cite{DMS}, while for a vector boundary we will employ the twisted Yangian introduced in \cite{MReg}. In order to distinguish these two algebras we name them the twisted Yangian of {\it type I} and of {\it type II} respectively.   
We note that reflection algebras and Yangians for $\gl(n)$ in various contexts have been extensively studied in \cite{MR,Mu,Doikou:2004hy}, for superalgebras $\sl(m|n)$ and $\gl(m|n)$ in \cite{BR,Arnaudon:2004sd,Doikou:2009dp,Doikou:2009kd}.
See \cite{MacKay:2011zs} for an `achiral' extension of such Yangians.

We also give some arguments that the quantum affine coideal subalgebras in the rational $q\to1$ limit specialize to the twisted Yangians. This is an important but quite technical question that is worth of exploration on its own, thus we will be rather concise and heuristic concerning this claim in the present paper. A thorough exploration of the Yangian limit of the quantum affine enveloping algebras can be found in \cite{Tolstoy,GTL,GM1}. 

This paper is organized as follows. In section 1 we give the necessary preliminaries. We recap the construction of quantum symmetric pairs and coideal subalgebras following closely to \cite{Le}. We then give the definition of the quantum affine coideal subalgebra, and also the definitions of the twisted Yangians introduced in \cite{DMS} and \cite{MReg}. In section 2 we construct reflection algebras for the quantum affine algebra $\sltwo$. In section 3 we construct the twisted Yangians for $\ytwo$. In sections 4 and 5 we repeat the same derivations for quantum groups $\glone$ and $\yone$ respectively. Appendix A contains a heuristic Yangian limit $\sltwo\to\ytwo$.\\

{\bf \noindent Acknowledgements.}
The author thanks Gustav Delius and Stefan Kolb for many valuable discussions, and Niall MacKay for his comments and reading of the manuscript. The author also thanks Marius de Leeuw, Takuya Matsumoto, Alessandro Torrielli and the UK EPSRC for funding under grant EP/H000054/1.


\section{Preliminaries}


\subsection{Quasitriangular Hopf algebras and the Yang-Baxter equation}
Let $\cA$ be a Hopf algebra 
over $\bb{C}$ equipped with multiplication $\mu : \cA \tc \cA \to \cA$, unit $\iota: \bb{C} \to \cA$, comultiplication $\Delta : A\to A\tc A$, counit $\epsilon : \cA \to \bb{C}$ and antipode $S : \cA \to \cA$. 
Let $\sigma : \cA \tc \cA \to \cA \tc A$ be the $\bb{C}$-linear map such that $\sigma(a_1 \tc a_2) = a_2 \tc a_1$ for any $a_1,\,a_2 \in \cA$. Then $(\cA,\mu^{op}, \iota, \Delta^{op}, \epsilon, S^{-1})$, where $\mu^{op} = \mu \circ \sigma$ and $\Delta^{op} = \sigma \circ \Delta$, is called the \textit{opposite} Hopf algebra of $\cA$ and denoted $\cA^{op}$. 

Let $\cA$ be a quasitriangular Hopf algebra. Then there exists an invertible element $\cR \in \cA \tc \cA$ called the \textit{universal} $R$-matrix 
such that
\begin{align} \label{UniInt}
\Delta^{op}(a) = \cR \, \Delta(a) \cR^{-1} \quad\text{for any} \quad a\in \cA\,, 
\end{align}
which satisfies the universal Yang-Baxter equation
\begin{align} \label{UniYBE}
\cR_{12} \cR_{13} \cR_{23} = \cR_{23} \cR_{13} \cR_{12} \,, 
\end{align}
where $\cR_{12} \in \cA \tc \cA \tc 1$, $\cR_{13} \in \cA \tc 1 \tc \cA $ and $\cR_{23} \in 1 \tc \cA \tc \cA$.

Let $\cA$ be a quantum affine universal enveloping algebra $\Ua$. Let $\cV$ be a finite dimensional vector space and $T_z:\cA\to End(\cV)$ be a finite dimensional representation of $\cA$, where $z$ denotes the spectral parameter of the representation. Then $(T_z\tc T_w):\cR\to R(z/w) \in End(\cV\tc\cV)$ maps the universal $R$-matrix to a matrix called the trigonometrical $R$-matrix. In such a way \eqref{UniInt} becomes the intertwining equation,
\begin{align} \label{Int}
(T_z \tc T_w)[\Delta^{op}(a)] \, R(z/w) = R(z/w) \, (T_z \tc T_w)[\Delta(a)]\,, 
\end{align}
and the Yang-Baxter equation \eqref{UniYBE} on the space $\cV\tc\cV\tc\cV$ becomes
\begin{align} \label{YBE}
R_{12}(z/w) R_{13}(z) R_{23}(w) = R_{23}(w) R_{13}(z) R_{12}(z/w) \,.
\end{align}
In the case of an irreducible representation $T_z \otimes T_w$ the intertwining equation \eqref{Int} defines the $R$-matrix uniquely up to an overall scalar factor. Furthermore, this $R$-matrix satisfies \eqref{YBE} automatically.


\vspace{0.1in}

\subsection{Coideal subalgebras and the reflection equation} 
Consider the reflection equation \cite{Sk}
\begin{align} \label{RE}
R_{12}(z/w) K_{13}(z) R_{12}(zw) K_{23}(w) = K_{23}(w) R_{12}(zw) K_{13}(z) R_{12}(z/w)
\end{align}
defined on the tensor space $\cV\tc\cV\tc\cW$, where $\cW$ is a boundary vector space. Here $K_{13}(z)$ and $K_{23}(w)$ are reflection matrices such that 
$1\tc K_{23}(w)\in \cV \tc End(\cV\tc\cW)$ and a similar relation holds for $K_{13}(z)$. 

Let $\cB\subset\cA$ be a left coideal subalgebra,
\begin{align}
\Delta(b) \in \cA\tc\cB \quad\text{for all} \quad b\in\cB \,.   
\end{align}
Let $\bar{T}_s:\cB\to End(\cW)$ be a finite dimensional representation of $\cB$, called the {\it boundary} representation; here $s$ denotes the boundary spectral parameter. 

\begin{defn}
A coideal subalgebra $\cB$ is called a quantum affine reflection algebra if the intertwining equation
\begin{align} \label{IntKV}
(T_{1/z}\tc\bar{T}_s)[\Delta(b)] \, K(z)  = K(z) \, (T_z\tc\bar{T}_s)[\Delta(b)] \quad\text{for all} \quad b \in \cB\,, 
\end{align}
for some representations $T_z$ and $\bar{T}_s$ defines a $K$-matrix $K(z)\in End(\cV\tc\cW)$ satisfying the reflection equation \eqref{RE}.
\end{defn}

Let the boundary vector space be one-dimensional, $\cW=\bb{C}$. Then $\bar{T}_s =\epsilon$ and $K(z)\in End(\cV\tc\bb{C})$. In this case the intertwining equation \eqref{IntKV} becomes
\begin{align} \label{IntKS}
(T_{1/z}\tc\epsilon)[\Delta(b)] \, K(z)  = K(z) \, (T_z\tc\epsilon)[\Delta(b)] \quad\text{for all} \quad b \in \cB\,.
\end{align}

For an irreducible representation $T_z$ (resp.\ $T_z\tc\bar{T}_s$) of $\cB$, the intertwining equation \eqref{IntKS} (resp. \eqref{IntKV}) defines the $K$-matrix uniquely up to an overall scalar factor. Note that the boundary representation $\bar{T}_s$ may be different from $T_z$; however, in this paper we will consider the $\bar{T}_s \cong T_z$ case only.

\begin{defn}
Let $\bar{T}_s \cong T_z$ be a non-trivial boundary representation. Then we call \eqref{IntKV} the intertwining equation for a vector boundary. We call \eqref{IntKS} the intertwining equation for the singlet boundary.
\end{defn}

For a given algebra $\cA$ and $R$-matrix $R(z)$ there can be a family of coideal subalgebras that define inequivalent solutions of the reflection equation. The most general solutions of the reflection equation are known for some Lie algebras only, and even less is known about the universal solution (see e.g.\ \cite{DKM}). In the following subsections we will consider coideal subalgebras compatible with the reflection equation for the quantum affine enveloping algebras and Yangians.


\vspace{0.1in}

\subsection{Quantum symmetric pairs and coideal subalgebras} 
We will begin by introducing the necessary notation and then we will give the definition of the coideal subalgebras of the universal enveloping algebras. We will be adhering closely to \cite{Le1}.


Let $\g$ be a semisimple Lie algebra of rank $n$. Let $\Phi$ denote the root space of $\g$, and $\Phi^{+}$ be the set of positive roots. Let $\pi=\{\alpha_i\}_{i\in\bf{I}}$ be a basis of the simple positive roots in $\Phi^+$. Here ${\bf I}=\{1,\ldots,n\}$ denotes the set of Dynkin nodes of $\g$. We will use $\lambda$ to denote any root in $\Phi$. Let $(\cdot,\cdot)$ denote a non-degenerate Cartan inner product on ${\bf h^*}$, the dual of the Cartan subalgebra ${\bf h}$ of $\g$. Then the matrix elements of the Cartan matrix $(a_{ij})_{i,j\in\bf{I}}$ are given by $a_{ij} = 2(\alpha_{i},\alpha_{j})/(\alpha_i,\alpha_i)$. There exists a set of coprime positive integers $(r_i)$ such that $(b_{ij})=(r_i a_{ij})$ is symmetric and is called the symmetrized Cartan matrix.

The triangular decomposition of $\g$ is given by ${\bf n^- \tp h \tp n^+}$, and the basis for $\bn^-$ (resp.\ $\bn^+$) is $\{f_{i}\}_{i\in\bf{I}}$ (resp.\ $\{e_{i}\}_{i\in\bf{I}}$). Let $h_i=[e_i,f_i]$ for all $i\in\bf{I}$. Then $\{e_i,f_i,h_i\}_{i\in\bf{I}}$ is a Chevalley basis for $\g$ satisfying
\begin{align} \label{rels:Lie}
[h_i,h_j]=0 \,, \qquad [e_i,f_j] = \delta_{ij} h_i \,,\qquad [h_i,e_j] = a_{ji} e_j \,,\qquad [h_i,f_j] = -a_{ji} f_j \,,
\end{align}
and the Serre relations
\begin{align} \label{rels:Serre}
(\ad\,e_i)^{1-a_{ji}}\,e_j =0 \,,\qquad (\ad\,f_i)^{1-a_{ji}}\,f_j =0 \,.
\end{align}

Let $\theta:\g\to\g$ be a maximally split involutive Lie algebra automorphism (involution) of $\g$, i.e.
\begin{align} \label{msi}
\theta(\bh) &=\bh \,, && \{\theta(e_i)=e_i,\,\theta(f_i)=f_i \,|\, \theta(h_i)=h_i\}\,, && \{\theta(e_i)\in\bn^-,\,\theta(f_i)\in\bn^+ \,|\, \theta(h_i)\neq h_i\}\,.
\end{align}
It defines a symmetric pair $(\g,\g^{\theta})$, where $\g^\theta$ is the $\theta$--fixed subalgebra of $\g$, and induces an involution $\Theta$ of the root space $\Phi$. Let $\pi_{\Theta}=\{\Theta(\alpha_i)=\alpha_i\,|\,\alpha_i\in\pi\}$ denote the $\Theta$--fixed subset of $\pi$. Then,  by \eqref{msi}, $\Theta(-\alpha_j)\in\Phi^{+}$ for all $\alpha_j\in\pi\backslash\pi_\Theta$.

Let $p$ be a permutation of $\{1,\ldots,n\}$ such that
\begin{align}
\Theta(\alpha_j) \in - \alpha_{p(j)} - \mathbb{Z} \pi_\Theta \quad \text{for all} \quad \alpha_{j}\notin\pi_{\Theta} \,, 
\end{align}
and $p(i)=i$ otherwise.  
Let $\pi^{*}$ be a maximal subset of $\pi\backslash\pi_{\Theta}$ such that $\alpha_{j}\in\pi^{*}$ if $p(j)=j$, or only one of the pair $\alpha_j$, $\alpha_{p(j)}$ is in $\pi^*$ if $p(j)\neq j$.  Then for a given $j$ such that $\alpha_{j}\in\pi^{*}$ there exists a sequence $\{\alpha_{j_{1}},\ldots,\alpha_{j_{r}}\}$, where $\alpha_{j_k}\in\pi_\Theta$, and a set of positive integers $\{m_{1},\ldots,m_{r}\}$ such that $\theta$ defined by 
\begin{align} \label{Inv:g1}
\theta(f_i) &= f_i \,,\qquad \theta(e_i) = e_i \,,\qquad \theta(h_i) = h_i \qquad \text{for all} \quad \alpha_i\in\pi_\Theta \,,
\end{align}
and
\begin{align} \label{Inv:g2}
\theta(f_j) &= \big( \ad\,e_{j_1}^{(m_1)}\cdots e_{j_r}^{(m_r)} \big) e_{p(j)} \,, 
& \theta(f_{p(j)}) &= (-1)^{m_{(j)}} \big( \ad\,e_{j_r}^{(m_r)}\cdots e_{j_1}^{(m_1)} \big)e_{j} \,,\el
\theta(e_j) &= (-1)^{m_{(j)}} \big( \ad\,f_{j_1}^{(m_1)}\cdots f_{j_r}^{(m_r)} \big) f_{p(j)} \,, 
& \theta(e_{p(j)}) &= \big( \ad\,f_{j_r}^{(m_r)}\cdots f_{j_1}^{(m_1)} \big) f_{j} \,,\el
\theta(h_j) &= - m_1 h_{j_1} - \ldots - m_r h_{j_r} - h_{p(j)}\,,
& \theta(h_{p(j)}) &= - m_1 h_{j_1} - \ldots + m_r h_{j_r} - h_{j}\,, 
\end{align}
for all $\alpha_j\in\pi^*$ is an involution of $g$ (up to a slight adjustment and rescaling of the definition of powers $(m_j)$ such that $[\theta(e_j),\theta(f_j)]=\theta(h_j)$ and
\begin{align}
\big( \ad\,f_{j_1}^{(m_1)}\cdots f_{j_r}^{(m_r)} \big) \big[ \big( \ad\,e_{j_r}^{(m_r)}\cdots e_{j_1}^{(m_1)} \big)\,e_{j} \big] &= e_{j} \,,\el
\big( \ad\,e_{j_1}^{(m_1)}\cdots e_{j_r}^{(m_r)} \big) \big[ \big( \ad\,f_{j_r}^{(m_r)}\cdots f_{j_1}^{(m_1)} \big) f_{j} \big] &= f_{j} \,,
\end{align}
would hold). Here $(\ad\,a) b = [a,b]$ and $m_{(j)}= m_1 + \ldots + m_r$. Note that the notation used in \eqref{Inv:g2} corresponds to $\Theta(\alpha_j) = - m_1 \alpha_{j_1} - \ldots - m_r \alpha_{j_r} - \alpha_{p(j)}$. The special case $\Theta(\alpha_j)=-\alpha_j$ for all simple roots $\alpha_j\in\pi$ gives Chevalley anti-automorphism $\kappa(f_j)=e_j$, $\kappa(e_j)=f_j$, $\kappa(h_j)=-h_j$.

\vspace{0.2cm}


Let the quantum deformed universal enveloping algebra $\Uq$ of a semisimple complex Lie algebra $\g$ of rank $n$ be generated by the elements $\xi_i^\pm,\; k_{i}^{\pm1}$ ($k_i=q^{r_i h_i}$, $i\in\bf{I}$, and $q\in\bb{C}^\times$ is transcendental) that correspond to the standard Chevalley-Serre realization satisfying
\begin{align} \label{rels:qLie}
& k_i k_i^{-1} = k_i^{-1} k_i = 1 \,, \qquad\quad k_i k_j = k_j k_i \,, \el
& k_i \xi^{\pm}_j k_i^{-1} = q^{\pm b_{ij}} \xi^{\pm}_j \,, \qquad\quad [\xi^+_i , \xi^-_j] = \delta_{ij} \frac{k_i-k_i^{-1}}{q_i-q_i^{-1}} \,,
\end{align}
and the quantum Serre relations
\begin{align} \label{rels:qSerre}
\sum_{m=0}^{1-a_{ij}} (-1)^m 
\left[\!\!\begin{array}{c} 1-a_{ij}\\m\\ \end{array}\!\!\right]_{q_i} \!
{(\xi^\pm_i)}^m \xi^{\pm}_j {(\xi^{\pm}_i)}^{1-a_{ij}-m}  = 0 \,, \quad \text{for all}\quad i\neq j \,.
\end{align}
The notation used in here is $q_i = q^{r_i}$ and
\begin{align} 
[n]_q = \frac{q^n-q^{-n}}{q-q^{-1}} \,, \qquad [n]_q! = [n]_q [n-1]_q \cdots [1]_q \,, \qquad \left[\!\!\begin{array}{c} n\\m\\ \end{array}\!\!\right]_q = \frac{[n]_q!}{[n-m]_q!\,[m]_q!} \,.
\end{align}

The algebra $\Uq$ becomes a Hopf algebra when equipped with 
the coproduct $\Delta$, antipode $S$ and counit $\epsilon$ given by
\begin{align}
\label{cop1}
& \Delta (k_i) \; = k_i \tc k_i \,, \hspace{-1cm} && S(k_i) \; = k_i^{-1}\,, \hspace{-1cm} && \epsilon (k_i) = 1\,, \el
& \Delta (\xi^+_i) = \xi^+_i \tc 1 + k_i \tc \xi^+_i\,, \hspace{-1cm} && S(\xi^+_i) = - k_i^{-1} \xi^+_i\,, \hspace{-1cm} && \el
& \Delta (\xi^-_i) = \xi^-_i \tc k_i^{-1} + 1 \tc \xi^-_i\,, \hspace{-1cm} && S(\xi^-_i) = - \xi^-_i k_i\,, \hspace{-1cm} && \epsilon(\xi^\pm_i) =0 \,.
\end{align}
Being a Hopf algebra, $\Uq$ admits a right adjoint action making
$\Uq$ into a right module. The right adjoint action is defined by
\begin{align}
\left(\adr \,\xi^{+}_i \right) a &= k_i^{-1} a\, \xi^{+}_i - k_i^{-1} \xi^{+}_i a \,, & \left(\adr \,\xi^{-}_i \right) a &= a\, \xi^{-}_i - \xi^{-}_i k_i \,a\, k_i^{-1} \,, 
& \left(\adr \,k_i \right) a &= k_i^{-1} a\, k_i \,.
\end{align}
We shall also be using a short-hand notation $\big(\adr \,\xi^{\pm}_i \cdots \xi^{\pm}_j \big) a = \big(\adr \,\xi^{+}_i \cdots \big(\adr \,\xi^{+}_j \big)\big) a$, for any $a\in\Uq$.

Let $\cT$ be an abelian subgroup $\cT\subset\Uq$ generated by $k_i^{\pm}$. Set $Q(\pi)$ to be equal to the integral lattice generated by $\pi$, i.e.\  $Q(\pi) = \sum_{1\leq i\leq n}\mathbb{Z}\alpha_{i}$. Then there is an isomorphism $\tau$ of abelian groups from $Q(\pi)$ to $\cT$ defined by $\tau(\alpha_{i})=k_{i}$, thus for every $\lambda\in\Phi$ there is an image $\tau(\lambda)\in \cT$.

\vspace{0.2cm}

Consider the involution $\theta$ of $\g$ defined in \eqref{Inv:g1} and \eqref{Inv:g2}. It can be lifted to the quantum case in the following sense. 

\begin{thrm} [Theorem 7.1 of \cite{Le}] \label{Trm:AutUq}
There exists an algebra automorphism $\tth$ of $\Uq$ such that
\begin{align} \label{AutUq}
& \tth(q) = q^{-1} \,, \el
& \tth(\xi^\pm_i) = \xi^\pm_i \quad\text{for all}\quad \alpha_i\in\pi_\Theta \,, \el
& \tth(\tau(\lambda)) = \tau(\Theta(-\lambda)) \quad\text{for all}\quad\tau(\lambda)\in\cT \,, \el
& \tth(\xi^{-}_j) = \big[ \big( \adr \,{\xi^{+}_{j_1}}^{(m_1)} \cdots {\xi^{+}_{j_r}}^{(m_r)} \big) k_{p(j)}^{-1} \xi^{+}_{p(j)} \big] \el
& \text{and}\quad \tth(\xi^{-}_{p(j)}) = (-1)^{m_{(j)}} \big[ \big( \adr \,{\xi^{+}_{j_r}}^{(m_r)} \cdots {\xi^{+}_{j_1}}^{(m_1)} \big) k_{j}^{-1} \xi^{+}_{j} \big] \quad\text{for all}\quad \alpha_j\in\pi^* \,.
\end{align}
\end{thrm}

This construction allows us to define a left coideal subalgebra of $\Uq$ induced by the involution $\Theta$.  Let $\cT_\Theta = \{\tau(\lambda)\,|\,\Theta(\lambda)=\lambda \}$ be a $\Theta$--fixed subalgebra of $\cT$. Let $\mathcal{M}$ be a Hopf subalgebra of $\Uq$ generated by $\xi^\pm_i$, $k_i^{\pm1}$ for all $\alpha_i\in\pi_\Theta$. Note that $k_j k_{p(j)}^{-1}\in\cT_\Theta$ for all $\alpha_j\in\pi^*$, thus $\cT_\cM\subseteq\cT_\Theta$ where $\cT_\cM=\{k_i^{\pm1}\}$ is the Cartan subgroup of $\cM$. Furthermore, $\tth^2=id$ when restricted to $\cM$ and to $\cT$. Finally, in the $q\to1$ limit $\tth$ specializes to $\theta$. 

\begin{thrm} [Theorem 7.2 of \cite{Le}] \label{Trm:BofUq}
The subalgebra $\cB\subset\Uq$ generated by $\cM$, $\cT_\Theta$ and the elements
\begin{align} \label{Bmj}
B^{-}_j = \xi^{-}_j k_j - d_j \, \tth(\xi^{-}_j) k_j \quad \text{for all}\quad\alpha_j\in\pi\backslash\pi_\Theta \,,
\end{align}
and suitable $d_j\in\bb{C}^\times$ is a left coideal subalgebra of $\Uq$.
\end{thrm}

Let $U^+$ (resp.\ $U^-$) be the subalgebra of $\Uq$ generated by $\xi^+_i$ (resp.\ $k_i \xi^-_i$) for all $\alpha_i\in\pi_\Theta$. Set $\cM^\pm=U^\pm\cap\cM$. By the definition, the elements $\tth(\xi^-_j)k_j$ are such that (see Section 6 and the proof of the Theorem 7.2 of \cite{Le}),
\begin{align} 
\Delta(\tth(\xi^-_j)k_j) \in k_j \tc \tth(\xi^-_j)k_j + \Uq \tc \cM^+ \cT_\Theta \subset \Uq \tc \cB \,.
\end{align}
Hence the coproducts of $B^{-}_j$ are of the following form,
\begin{align} \label{Cop:Bmj}
\Delta(B^{-}_j) \in k_j \tc B^{-}_j + \Uq \tc \cM^+ \cT_\Theta \subset \Uq \tc \cB \,.
\end{align}

\begin{crl} \label{Crl:1}
The subalgebra $\mathcal{D}\subset\cB$ generated by $\cM$, $\cT_\Theta$ and the elements $B^{-}_j$ for any but not all $\alpha_j\in\pi\backslash\pi_\Theta$ is a left coideal subalgebra of $\Uq$.
\end{crl}

The pair ($\Uq$, $\cB$) is called the quantum symmetric pair and is the quantum analog of the pair of enveloping algebras ($\cU(\g)$, $\cU(\g^\theta)\,$). For more details consult Section 7 of \cite{Le}; for explicit $\cB$'s for various simple Lie algebras see \cite{Le1}. 
Note that the action of $\tth$ on $\xi^{+}_j$ is not explicitly defined by Theorem \ref{Trm:AutUq}, but is constrained by requiring $\tth$ to be an automorphism of $\Uq$. 

\vspace{0.2cm}

In some cases it is more convenient to work with an equivalent coideal subalgebra $\cB'$ which is obtained by interchanging all $\xi^-_i$ and $k_i^{-1} \xi^+_i$, $\{i\in{\bf I}\}$. Let us show this explicitly.
Consider a $\bb{C}$-linear algebra anti-automorphism $\kappa_B$ of $\Uq$ given by
\begin{align} \label{kb}
\kappa_B(\xi^{-}_i) = c_B^{-1} k_i^{-1} \xi^{+}_i \,, \quad \kappa_B(\xi^{+}_i) = c_B \, \xi^{-}_i k_i \,, \quad \kappa_B(\tau(\lambda)) = \tau(\lambda) \,.
\end{align}
Then there exists $c_B\in\bb{C}^\times$ such that $\kappa_B(\cB)=\cB$ holds. This is easy to check. Firstly, 
\begin{align}
\kappa_B((\adr\,\xi^+_i) b ) = - c_B (\adr\,\xi^-_i) \kappa_B(b) \,,\qquad \kappa_B((\adr\,\xi^-_i) b ) = - c^{-1}_B (\adr\,\xi^+_i) \kappa_B(b) \,.
\end{align}
Recall that
\begin{align}
\big(\adr\,{\xi^{-}_{j_1}}^{(m_1)} \cdots {\xi^{-}_{j_r}}^{(m_r)}\big) \big(\adr\,{\xi^{+}_{j_r}}^{(m_r)} \cdots {\xi^{+}_{j_1}}^{(m_1)}\big) k_j^{-1} \xi^+_j = k_j^{-1} \xi^+_j \,.
\end{align}
This gives
\begin{align}
\kappa_B(B^{-}_{j}) &= c_B^{-1} \Big( q^{a_{jj}} k_j^{-1} \xi^{+}_j k_j - d'_j \big[\big(\adr\,{\xi^{-}_{j_1}}^{(m_1)} \cdots {\xi^{-}_{j_r}}^{(m_r)}\big) \xi^{-}_{p(j)}\big] k_j \Big) \el
 &= c_B^{-1} q^{a_{jj}} d'_j (-1)^{m_{(j)}} \big[ \big(\adr\,{\xi^{-}_{j_1}}^{(m_1)} \cdots {\xi^{-}_{j_r}}^{(m_r)}\big) B^{-}_{p(j)} \big] \, k_{p(j)}^{-1} k_j \,,
\end{align}
where $d'_j = d_j c_B^{m_{(j)}+2} q^{\sum_k m_k a_{jj_k}-a_{jj}}$, and we have required $d'_j = q^{a_{jj}} d_{p(j)}^{-1}$. Thus $\kappa_B(B^-_j)\in \cB$, and in a similar way one could show that $\kappa_B(B^-_{p(j)})\in \cB$. Finally, the property is manifest for $\cM$ and of $\cT_\Theta$. This implies that one can replace all generators $B^{-}_j$ \eqref{Bmj} by an equivalent set of generators $B^{+}_j$ that are obtained by interchanging all $\xi^{-}_i$ and $k^{-1}_i \xi^{+}_i$ ($i\in\bf{I}$) in \eqref{AutUq} and \eqref{Bmj} giving \cite{Kolb}
\begin{align} \label{Bpj}
B^{+}_j = \xi'^{+}_j k_j - d_j \, \tth'(\xi'^{+}) k_j \quad \text{for all} \quad \alpha_j\in\pi\backslash\pi_\Theta \,,
\end{align}
and suitable $d_j\in\bb{C}^\times$; here $\xi'^{+}_j=k_j^{-1} \xi^{+}_j$ and $\tth'$ is defined by
\begin{align} \label{tth:primed}
\tth'(\xi'^{+}_j) &= (-1)^{m_{(j)}}\big[ \big( \adr \,{\xi^{-}_{j_1}}^{(m_1)} \cdots {\xi^{-}_{j_r}}^{(m_r)} \big) \xi^{-}_{p(j)} \big] \,, \el
\tth'(\xi'^{+}_{p(j)}) &= \big[ \big( \adr \,{\xi^{-}_{j_r}}^{(m_r)} \cdots {\xi^{-}_{j_1}}^{(m_1)} \big) \xi^{-}_{j} \big] \,.
\end{align}
The coproducts of $B^{+}_j$ have the following form
\begin{align} \label{Cop:Bpj}
\Delta(B^{+}_j) \in k_j \tc B^{+}_j + \Uq \tc \cM^- \cT_\Theta \,.
\end{align}
This leads to the following corollaries:

\begin{crl} \label{Crl:2}
There exists an algebra automorphism $\tth'$ of $\Uq$ such that
\begin{align}
& \tth'(q) = q^{-1} \,, \el
& \tth'(\xi^\pm_i) = \xi^\pm_i \quad\text{for all}\quad \alpha_i\in\pi_\Theta \,, \el
& \tth'(\tau(\lambda)) = \tau(\Theta(-\lambda)) \quad\text{for all}\quad\tau(\lambda)\in\cT \,, 
\end{align}
and \eqref{tth:primed} holds for all $\alpha_j\in\pi^*$. It is an involution $\tth'^2=id$ when restricted to $\cM$ and to $\cT$. In the $q\to1$ limit $\tth'$ specializes to $\theta$.
\end{crl}
\begin{crl} \label{Crl:3}
The subalgebra $\cB'\subset\Uq$ generated by $\cM$, $\cT_\Theta$ and the elements
\begin{align} \label{Bja}
B^{+}_j = \xi'^{+}_j k_j - d'_j \, \tth(\xi'^{+}_j) k_j \quad \text{for all}\quad\alpha_j\in\pi\backslash\pi_\Theta \,,
\end{align}
and suitable $d'_j\in\bb{C}^\times$ is a left coideal subalgebra of $\Uq$.
\end{crl}

\begin{crl} \label{Crl:4}
The subalgebra $\mathcal{D'}\subset\cB'$ generated by $\cM$, $\cT_\Theta$ and the elements $B^{+}_j$ for any but not all $\alpha_j\in\pi\backslash\pi_\Theta$ is a left coideal subalgebra of $\Uq$.
\end{crl}

Note that in this case the action of $\tth'$ on $\xi^-_j$ is not explicitly defined, but is constrained requiring $\tth'$ to be an automorphism of $\Uq$.


\vspace{0.1in}

\subsection{Quantum affine coideal subalgebras}
We will further be interested in two particular extensions of the coideal subalgebras defined above. We will construct coideal subalgebras of the quantum affine algebra $\Ua$ that are associated with singlet and vector boundaries. 

Let $\hat\g$ be the (untwisted) affine Kac--Moody algebra. Let $(\wh{a}_{ij})_{i,j\in\hat{\bf{I}}}$ denote the extended Cartan matrix and $(\wh{b}_{ij})=(\wh{r}_i \, \wh{a}_{ij})$ be the symmetrized extended Cartan matrix. Here $\hat{\bf I}=\{0,1,\ldots,n\}$ denotes the set of Dynkin nodes of $\hat{\g}$. The set of the simple positive roots is given by $\wh\pi = \alpha_0\cup\pi$, where $\alpha_0$ is the affine root.
Recall that $\hat\g$ is an one--dimensional central extension of the Lie algebra $\cL(\g)=\g[z,z^{-1}]$ of Laurent polynomial maps $\bb{C}^\times \to \g$ under pointwise operations. The triangular decomposition is given by $\hat\g = \hat\bn^+\tp\hat\bh\tp\hat\bn^-$, where
\begin{align}
\hat\bn^\pm = z^{\pm1}\,\bb{C}[z^{\pm1}]\tc(\bn^{\mp}\tp\bh)\tp\bb{C}[z^{\pm1}]\tp\bn^{\pm} \,, \qquad\quad \hat\bh = (1\tc\bh)\tp\bb{C}K\tp\bb{C}D \,,
\end{align}
Here $K$ is the central element and $D$ is the derivation of the algebra. The Chevalley generators are given by
\begin{align}
E^+_i &= 1\tc e_i \,, \hspace{-1.5cm} & E^+_0 &= z \tp e_0 \in z\tp\bn^-\subset\hat\bn^+ \,, \el
E^-_i &= 1\tc f_i \,, \hspace{-1.5cm} & E^-_0 &= z^{-1} \tp f_0 \in z^{-1}\tp\bn^+\subset\hat\bn^- \,, \el
H_i\; &= 1\tc h_i \,, \hspace{-1.5cm} & H_0\, &= [E^+_0,E^-_0] \in [e_0, f_0] + \bb{C}K \subset \hat\bh \,,
\end{align}
where $e_0\in\g_{-\vartheta}$, $f_0\in\g_{\vartheta}$ are such that $\vartheta\in\Phi^+$ is the highest root of $\g$. 

The elements $E^\pm_i$, $H_i$ ($i\in{\bf \hat{I}}$) generate a subalgebra $\tilde\g\subset\hat\g$ such that $\hat\g=\tilde\g\tp\bb{C}D$ is a semi--direct product Lie algebra. The derivation $D=z\,d/dz$ of $\bb{C}[z,z^{-1}]$ acts on $\tilde\g$ by
\begin{align} \label{D_action}
[D,E^\pm_0] = \pm E^\pm_0 \qquad\text{and}\qquad 
[D,H_0]=[D,H_i]=[D,E_i^\pm]=0 \quad\text{for all}\quad i\in{\bf I} \,.
\end{align}
Set $\delta\in\hat\bh^*$ such that $\delta(D)=1$ and $\delta(\bh\tp\bb{C}K)=0$. Then the affine root is given by $\alpha_0 = \delta - \vartheta$. 

Consider an involution $\theta$ of $\tilde\g$ such that the associated root space involution $\Theta$ is given by
\begin{align} \label{AffInv}
\Theta (\alpha_0) \in -\alpha_{p(0)} - \bb{Z} (\pi\backslash\alpha_{p(0)}) \qquad\text{and}\qquad \Theta (\alpha_i) = \alpha_i \;\; \text{for all} \;\; \alpha_i \in \pi\backslash\alpha_{p(0)}\,,
\end{align}
and satisfying the following constraint,
\begin{align}
\alpha_0-\Theta(\alpha_0) = k \delta\,, \qquad\text{where}\qquad \begin{cases}
k=1 \text{ for } p(0)\neq0 \,,\\ k=2 \text{ for } p(0)=0 \,,\end{cases}
\end{align}
here $p(0)\in\{0,\ldots,n\}$, and $\pi\backslash\alpha_{p(0)}=\pi$ if $p(0)=0$. Define $\theta(D)=-D$. Then, for the $p(0)=0$ case, the relations
\begin{align} \label{theta_D}
[\theta(D),\theta(E^\pm_0)]=\theta([D,E^\pm_0]) \qquad\text{and}\qquad [\theta(D),\theta(E^\pm_i)]=\theta([D,E^\pm_i]) \quad\text{for all}\quad i\in{\bf I}
\end{align}
are satisfied, and thus the involution $\theta$ can be naturally lifted to an involution of $\hat\g$. Otherwise, if $p(0)\neq0$, relations \eqref{theta_D} do not hold and such lift is not possible. Nevertheless, $\theta^2=id$ on $\hat\g$ for both cases.

\vspace{0.2cm}

Let $\Ua$ be the universal enveloping algebra of $\hat\g$. The algebra $\Ua$ in the standard Drinfeld-Jimbo realization is generated by the elements $\xi^\pm_i$, $k_i^{\pm}$ ($i\in\hat{\bf{I}}$) satisfying \eqref{rels:qLie} and \eqref{rels:qSerre} with $a_{ij}$ (resp.\ $b_{ij}$) replaced by $\wh{a}_{ij}$ (resp.\ $\wh{B}_{ij}$). The subalgebra of $\Ua$ generated by $\xi^\pm_i$, $k^{\pm1}_i$ ($i\in{\bf I}$) is a Hopf subalgebra and is isomorphic as a Hopf algebra to $\Uq$. In this way, the modules of $\Ua$ restrict to the modules of $\Uq$ \cite{CPReps}. 
The involution $\theta$ defines a Hopf subalgebra $\mathcal{M}\subset\Ua$ generated by $\xi^\pm_i$, $k_i^{\pm1}$ for all $\alpha_i\in\pi\backslash\alpha_{p(0)}$ and an abelian subgroup $\cT_\Theta$ in the sense as described above. 
Furthermore, the involution $\theta$ induces an automorphism of $\Ua$ in the following way. 

\begin{prop} [Theorem \ref{Trm:AutUq} for the quantum affine algebras] \label{Con:AutUa}
Let a root space involution $\Theta$ be defined as in \eqref{AffInv}. Then there exists a sequence $\{\alpha_{0_1},\ldots,\alpha_{0_r}\}$, where $\alpha_{0_k}\in\pi\backslash\alpha_{p(0)}$, and a set of positive integers $\{m_{1},\ldots,m_{r}\}$ such that the algebra map $\tth$ defined by
\begin{align} \label{AutUa}
& \tth(q) = q^{-1} \,, \el
& \tth(\xi^\pm_i) = \xi^\pm_i \quad\text{for all}\quad \alpha_i\in\pi\backslash\alpha_{p(0)} \,, \el
& \tth(\tau(\lambda)) = \tau(\Theta(-\lambda)) \quad\text{for all}\quad\tau(\lambda)\in\cT \,, \el
& \tth(\xi^{-}_0) = \big[ \big( \adr \,{\xi^{+}_{0_1}}^{(m_1)} \cdots {\xi^{+}_{0_r}}^{(m_r)} \big) k_{p(0)}^{-1} \xi^{+}_{p(0)} \big] \el
& \text{and}\quad \tth(\xi^{-}_{p(0)}) = (-1)^{m_{(0)}} \big[ \big( \adr \,{\xi^{+}_{0_r}}^{(m_r)} \cdots {\xi^{+}_{0_1}}^{(m_1)} \big) k_{0}^{-1} \xi^{+}_{0} \big] \,,
\end{align}
can be extended to an automorphism of $\Ua$. Furthermore, it is an involution $\tth^2=id$ when restricted to $\cM$ and to $\cT$. In the $q\to1$ limit $\tth$ specializes to $\theta$.
\end{prop}

Note that for $p(0)=0$ case the last two lines of \eqref{AutUa} are equivalent. The proof of this proposition would be a lift of the proof of the Theorem 7.1 of \cite{Le}. This is because the sequence $\{m_{1},\ldots,m_{r}\}$ does not include the affine root, which makes the whole construction very similar to the non-affine case. However here we will not attempt to give a proof as it goes beyond of the scope of the present work. We will concentrate on the quantum affine coideal subalgebras $\cB\subset\Ua$ compatible with the reflection equation. Set $\hat\pi^* = \{\alpha_0, \alpha_{p(0)} \}$ if $p(0)\neq0$, and $\hat\pi^* = \{\alpha_0\}$ otherwise. Then:

\begin{prop} \label{Trm:BofUa}

The algebra $\cB$ generated by $\mathcal{M}$, $\cT_\Theta$, and the elements
\begin{align} \label{Bma}
B^{-}_j = \xi^-_j k_j - d_j \tth(\xi^-_{j}) k_j \qquad\text{for}\qquad \alpha_j\in\hat\pi^*\,,
\end{align}
and suitable $d_j\in\bb{C}^\times$ is a quantum affine coideal subalgebra of $\Ua\,$. 
\end{prop}

\begin{proof}
The proof of this proposition is a direct lift of the proof of the Theorem \ref{Trm:BofUq} (Theorem 7.2 of \cite{Le}). We need to check that 
\begin{align}
\Delta(b) \in \Ua \tc \cB \qquad \text{for all}\qquad b\in\cB \,.
\end{align}
This property is manifest for all generators of $\cM$ and of $\cT_\Theta$. Next we need to show that it also holds for $B^{-}_j$. By the definition $\tth(\xi^-_j)$ is such that
\begin{align}
\Delta(\tth(\xi^-_j)k_j) \in k_j \tc \tth(\xi^-_j)k_j + \Ua \tc \cM^{+} \cT_\Theta \,.
\end{align}
Hence
\begin{align}
\Delta(B^{-}_j) \in k_j \tc B^{-}_j + \Ua \tc \cM^{+} \cT_\Theta \subset \Ua \tc \cB \,,
\end{align}
as required.
\end{proof}

Analogously to the non-affine case, we could instead introduce an equivalent coideal subalgebra $\cB'$ which is obtained by interchanging all $\xi^-_i$ and $k_i^{-1}\xi^+_i$, $\{i\in{\bf\hat{I}}\}$. This leads to the following corollaries:

\begin{crl}
There exists an algebra automorphism $\tth'$ of $\Ua$ such that
\begin{align}
& \tth'(q) = q^{-1} \,, \el
& \tth'(\xi^\pm_i) = \xi^\pm_i \quad\text{for all}\quad \alpha_i\in\pi\backslash\alpha_{p(0)} \,, \el
& \tth'(\tau(\lambda)) = \tau(\Theta(-\lambda)) \quad\text{for all}\quad\tau(\lambda)\in\cT \,, \el
& \tth'(\xi'^{+}_0) = (-1)^{m_{(0)}} \big[ \big( \adr \,{\xi^{-}_{0_1}}^{(m_1)} \cdots {\xi^{-}_{0_r}}^{(m_r)} \big) \xi^{-}_{p(0)} \big] \el
& \text{and}\quad \tth'(\xi'^{+}_{p(0)}) = \big[ \big( \adr \,{\xi^{-}_{0_r}}^{(m_r)} \cdots {\xi^{-}_{0_1}}^{(m_1)} \big) \xi^{-}_{0} \big] \,,
\end{align}
Furthermore, it is an involution $\tth'^2=id$ when restricted to $\cM$ and to $\cT$. In the $q\to1$ limit $\tth'$ specializes to the involution $\theta$.
\end{crl}

\begin{crl} \label{Crl:1a}
The subalgebra $\cB'\subset\Ua$ generated by $\cM$, $\cT_\Theta$ and the elements
\begin{align}
B^{+}_j = \xi'^{+}_j k_j - d'_j \, \tth(\xi'^{+}_j) k_j \qquad\text{for}\qquad \alpha_j\in\hat\pi^*\,,
\end{align}
and suitable $d'_j\in\bb{C}^\times$ is a left coideal subalgebra of $\Ua$. 
\end{crl}

\begin{defn}
We call $B^\pm_j$ the twisted affine generators.
\end{defn}

The coideal subalgebra $\cB$ (resp.\ $\cB'$) defined above with suitable $d_j\in\bb{C}^\times$ (resp.\ $d'_j\in\bb{C}^\times$) is a quantum affine reflection algebra. These two algebras are isomorphic by the construction.
In the cases when it is obvious from the context we will further refer to a quantum affine reflection algebra simply as a reflection algebra. The parameters $d_j$ (resp. $d'_j$) are constrained by solving the intertwining equation \eqref{IntKV} for all generators of $\cB$ (resp. $\cB'$). 

In certain cases, in particular for $p(0)\neq0$, the requirement for $d_j$ (or equivalently for $d'_j$) to be non-zero is too restrictive. Thus in such cases it can be more convenient to deal with a reflection algebra defined in the following remark.
\begin{remark} \label{Conj4a}
The subalgebra $\cB\subset\Ua$ generated by $\cM$, $\cT_\Theta$ and the elements
\begin{align} \label{Bpm}
B^{-}_0 = \xi^-_0 k_0 - d_- \tth(\xi^-_{0}) k_0 \qquad\text{and}\qquad 
B^{+}_0 = \xi'^{+}_0 k_0 - d_+ \tth'(\xi'^{+}_0) k_0
\end{align}
with suitable $d_\pm\in\bb{C}$ is a quantum affine reflection algebra. 
\end{remark}

We will give explicit examples supporting the claims above for both singlet and vector boundaries in the following sections. In Section 2 we will construct coideal subalgebras for $\sltwo$, and in Section 4 for $\glone$.

Note that for the $p(0)=0$ case inclusion of both $B_0^+$ and $B_0^-$ could potentially lead to an unwanted growth of $\cB$. This can be avoided by a suitable choice of $d_\pm$. Let us show this explicitly. Consider the following element,
\begin{align} \label{consistency}
& - d_- \big[ \big( \adr \,{\xi^{+}_{0_1}}^{(m_1)} \cdots {\xi^{+}_{0_r}}^{(m_r)} \big) B_0^+ k_\vartheta \big] \el
& \qquad \quad = \big[ - d_- \tth(\xi^{-}_{0}) + d_- d_+ \big( \adr \,{\xi^{+}_{0_1}}^{(m_1)} \cdots {\xi^{+}_{0_r}}^{(m_r)} {\xi^{-}_{0_r}}^{(m_r)} \cdots {\xi^{-}_{0_1}}^{(m_1)} \big) \xi^{-}_0 \big] k_0 k_\vartheta = B_0^- k_\vartheta \,,
\end{align}
where $k_\vartheta = k_1 \cdots k_n \in \cT_\Theta$. The last equality holds for a suitable choice of $d_\pm$. 

Finally note that one could equivalently choose $\xi'^{+}_j = \xi^{+}_j k_j^{-1}$ in \eqref{Bpm}. This would introduce a factor of $q^{\pm \wh{a}_{jj}}$ for $d_\pm$.


\vspace{0.1in}

\subsection{Yangian}
The Yangian $\Y(\g)$ of a Lie algebra $\g$ is a deformation of the universal
enveloping algebra of the polynomial algebra $\g[u]$. It is generated by the level-zero $\g$ generators $j^{a}$ and the level-one Yangian generators $\wh{j}^{a}$. Their commutators have the generic form
\begin{align}
[j^{a},j^b\,] = f_{\sk c}^{ab}\,j^{c}, \qquad\quad \bigl[j^{a},\wh{j}^b\bigr]=f_{\sk c}^{ab}\,\wh{j}^{c},
\end{align}
and are required to obey Jacobi and Serre relations
\begin{align}
\bigl[j^{[a},\bigl[\,j^b,j^{c]}\bigr]\bigr] = 0\,, \qquad\quad \bigl[\,\wh{j}^{[a},\bigl[\,\wh{j}^b,j^{c]}\bigr]\bigr] = {\alpha^2}\, a^{abc}_{\sk def} \,j^{\{d}j^{e}j^{f\}} \,, \label{Serre}
\end{align}
where $^{[a\,b\,c]}$ denotes cyclic permutations, $^{\{d\,e\,f\}}$ is the total symmetrization, and $a^{abc}_{\sk def} = \frac{1}{24} f_{\sk d}^{ag}\,f_{\sk e}^{bh}\,f_{\sk f}^{ck}\,f_{ghk}$. For $\g=\sl(2)$ the second equation in \eqref{Serre} is trivial and
\begin{align}
\big[\big[\,\wh{j}^a,\wh{j}^b\big],\big[j^l,\wh{j}^m\big]\big] 
 + \big[\big[\,\wh{j}^l,\wh{j}^m\big],\big[j^a,\wh{j}^b\big]\big]=\alpha^2(a^{abc}_{\sk deg} f^{lm}_{\sk c} + a^{lmc}_{\sk deg} f^{ab}_{\sk c})\,\wh{j}^{\{d} {j}^e \,j^{g\}} \label{SerreSL2}
\end{align}
needs to be used instead.
The indices of the structure constants $f_{\sk d}^{ab}$ are lowered
by means of the Killing--Cartan form $g_{bd}$. Here $\alpha$ 
is a formal level-one deformation parameter which is used to count the formal level of the algebra elements. In such a way the left and right hand sides of the expressions in \eqref{Serre} and \eqref{SerreSL2} are of the same level. 

The Hopf algebra structure is then equipped with the following coproduct $\Delta$, antipode $S$ and counit $\epsilon$,
\begin{align} \label{Hopf-Y}
\Delta (j^{a}) &= j^{a}\tc1+1\tc j^{a}, & S(j^a) &=-j^a \,, & \epsilon(j^a) &= 0 \,, \el
\Delta(\wh{j}^{a}) &= \wh{j}^{a}\tc1+1\tc\wh{j}^{a}+\frac{\alpha}{2}f_{\; bc}^{a}j^b\tc j^{c} \,, & S(\,\wh{j}^a) &= -\wh{j}^a - \frac{c_\g}{4} j^a \,, & \epsilon(\wh{j}^a) &= 0 \,,
\end{align}
where $c_{\g}$ is the eigenvalue of the quadratic Casimir operator in the adjoint representation ($f_{a}^{\;bc}f_{cbd}=c_{\g}\,g_{ad}$) and is required to be non-vanishing.

The finite-dimensional representations of $\Y(\g)$ are
realized in one-parameter families, due to the `evaluation automorphism'
\begin{align} \label{shift}
\tau_{u}:\Y(\g)\rightarrow\Y(\g)\qquad\; j^{a}\mapsto j^{a}\,,\qquad\wh{j}^{a}\mapsto\wh{j}^{a}+u j^{a} \,,
\end{align}
corresponding to a shift in the polynomial variable. On (the limited
set of) finite-dimensional irreducible representations of $\g$
which may be extended to representations of $\Y(\g)$, 
these families are explicitly realized via the `evaluation map'
\begin{align}
\mathrm{ev}_{u}:\Y(\g)\rightarrow\mbox{U}(\g)\qquad\; j^{a}\mapsto j^{a}\,,\qquad\wh{j}^{a}\mapsto uj^{a}\,,\label{ev_map}
\end{align}
which yields `evaluation modules' and $u$ is the spectral parameter. 

\vspace{0.2cm}

The level-two Yangian generators may be obtained by commuting level-one generators as
\begin{align}
\dwh{j}{}^{a} &= \frac{1}{c_{\g}}f^{a}_{\;\,bc}\,[\,\wh{j}^{c},\wh{j}^b] \,, \qquad\text{and}\qquad [\,\wh{j}^{a},\wh{j}^b]=f^{ab}_{\sk c}\,\dwh{j}{}^{c}+ X^{ab} ,\label{J2}
\end{align}
where the non-zero extra term $X^{ab}$ is constrained by the Serre relations
(\ref{Serre}) to satisfy $f^{[ab}_{\quad d} X^{c]d}=Y^{abc}$, here $Y^{abc}$ 
is the right hand side of the second equation in (\ref{Serre}) (and thus a fixed cubic combination of level-zero generators), and by \eqref{J2} to satisfy $f^a_{\;\;bc} X^{bc}=0$  \cite{MacKay:2004tc}.

\vspace{0.2cm}

Let $\cV$ be a finite dimensional vector space and $T_u:\Yg\to End(\cV)$ be an evaluation representation of $\Yg$ on $\cV$. Then $(T_u\tc T_v):\cR\to R(u-v) \in End(\cV\tc\cV)$ maps the universal $R$-matrix to a matrix, called rational $R$-matrix. In such a way \eqref{UniInt} becomes the intertwining equation,
\begin{align} \label{IntY}
R(u-v) \, (T_u \tc T_v)[\Delta(a)] = (T_u \tc T_v)[\Delta^{op}(a)] \, R(u-v)  \quad \text{for all}\quad a\in\Yg\,,
\end{align}
and equivalently the Yang-Baxter equation \eqref{UniYBE} becomes
\begin{align} \label{YBEY}
R_{12}(u-v) R_{13}(u) R_{23}(v) = R_{23}(v) R_{13}(u) R_{12}(u-v) \,.
\end{align}
%


\vspace{0.1in}

\subsection{Generalized twisted Yangians and the reflection equation}
The reflection equation for the Yangian algebra is obtained from \eqref{RE} in the same way as \eqref{YBEY} from \eqref{YBE} giving
\begin{align} \label{REY}
R_{12}(u-v) K_{13}(u) R_{12}(u+v) K_{23}(v) = K_{23}(v) R_{12}(u+v) K_{13}(v) R_{12}(u-v).
\end{align}
Let $\cB\subset\Yg$ be a left coideal subalgebra,
\begin{align}
\Delta(b) \in \Yg\tc\cB \quad\text{for all} \quad b\in\cB \,.   
\end{align}
Let $\bar{T}_s:\cB\to End(\cW)$ denote an evaluation representation of $\cB$ on the boundary vector space $\cW$. Here we assume $\cW$ to be finite dimensional.

\begin{defn}
The coideal subalgebra $\cB$ is called a Yangian reflection algebra if the intertwining equation
\begin{align} \label{IntKVY}
(T_{-u}\tc \bar{T}_{s})[\Delta(b)] \, K(u)  = K(u) \, (T_u\tc \bar{T}_{s})[\Delta(b)] \quad\text{for all} \quad b \in \cB\,, 
\end{align}
defines a $K$-matrix $K(u)\in End(\cV\tc\cW)$ satisfying reflection equation \eqref{REY}.
\end{defn}

Let the boundary vector space be one-dimensional, $\cW=\bb{C}$. Then $\bar{T}_s =\epsilon$ and $K(u)\in End(\cV\tc\bb{C})$. In this case the intertwining equation \eqref{IntKVY} becomes
\begin{align} \label{IntKSY}
(T_{-u}\tc\epsilon)[\Delta(b)] \, K(z)  = K(z) \, (T_u\tc\epsilon)[\Delta(b)] \quad\text{for all} \quad b \in \cB\,.
\end{align}
Note that \eqref{IntKVY} and \eqref{IntKSY} are Yangian equivalents of the quantum affine intertwining equations \eqref{IntKV} and \eqref{IntKS}. Finally, for an irreducible representation $T_u$ (resp.\ $T_u\tc\bar{T}_s$) of $\cB$ the intertwining equation \eqref{IntKSY} (resp.\ \eqref{IntKVY}) defines a $K$-matrix uniquely up to an overall scalar factor. As in the quantum affine case, the boundary representation $\bar{T}_s$ may be different from $T_u$. Here we will consider the $\bar{T}_s \cong T_u$ case only.

\vspace{0.2cm}


We will next identify two types of coideal subalgebras of $\Yg$ that are compatible with the reflection equation. These are the so-called (generalized) twisted Yangians introduced in \cite{DMS} and \cite{MReg}, and are constructed by defining involutions of $\Yg$ and requiring the coideal property to be satisfied. We will simply be calling these twisted Yangians to be of type I and type II respectively. We give the next two propositions without the proofs as they are straightforward.

\begin{prop}
Let a subalgebra $\h\subset\g$ be such that the splitting $\g=\h\oplus\m$ forms a symmetric pair
\begin{align}
\left[\h,\h\right]\subset\h,\qquad\left[\h,\m\right]\subset\m,\qquad\left[\m,\m\right]\subset\h\,.\label{symmetric_pair}
\end{align}
This splitting allows us to introduce an involution $\theta$ of $\g$ such that
\begin{align}
\theta(j^i) = j^i\,,\qquad \theta(j^p) = -j^p\,, \qquad \text{where}\qquad j^i\in\h \,,\; j^p\in\m\,.
\end{align}
Then $\theta$ can be extended to the involution $\bth$ of $\Yg$ such that
\begin{align}
\bth(\,\wh{j}^i) = -\wh{j}^i\,,\qquad \theta(\,\wh{j}^p) = \wh{j}^p\,, \qquad \theta(\alpha) = -\alpha\,.
\end{align}
\end{prop}

\begin{prop}
Let $\theta$ be a trivial involution of $\g$, 
\begin{align}
\theta(j^a) = j^a\,,\qquad  j^a\in\g\,.
\end{align}
Then it can be extended to a non-trivial involution $\bth$ of $\Yg$ such that
\begin{align}
\bth(\,\wh{j}^a) = -\wh{j}^a\,,\qquad \theta(\,\dwh{j}{}^a) = \dwh{j}{}^a\,, \qquad \theta(\alpha) = -\alpha\,.
\end{align}
\end{prop}

Involution $\bth$ endows $\Yg$ with the structure of a filtered algebra which combined with the requirement for the coideal property
\begin{align} \label{coidealY}
\Delta\left(\bth\big(\Yg\big)\right) \subset \Yg \tc \bth\big(\Yg\big) \,,
\end{align}
to be satisfied defines the twisted Yangian $\Ygh$. Here $\bth\big(\Yg\big)$ denotes the $\bth$--fixed subalgebra of $\Yg$, and $\theta(\g)$ denotes the $\theta$--fixed subalgebra of $\g$.

%
\begin{defn}
Let $\h=\theta(\g)$ be a non-trivial $\theta$--fixed subalgebra of $\g$. Then the twisted Yangian $\Y(\g,\h)$ of type I is a left coideal subalgebra of $\Y(\g)$ generated by the level-zero generators $j^{i}$ and the twisted level-one generators \cite{DMS,MacKay:2002at}.
\begin{align} \label{twistI}
\wt{j}^{p} &= \wh{j}^{p} + \alpha\, t\, j^p+\frac{\alpha}{4}f_{\;\, qi}^{p}\left(j^{q}\,j^{i}+j^{i}\,j^{q}\right),
& \Delta(\wt{j}^{p}) &= \wt{j}^{p}\tc1+1\tc\wt{j}^{p}+\alpha f_{\;\, qi}^{p}\,j^{q}\tc j^{i}\,,
\end{align}
where $i(,j,k,...)$ run over the $\h$-indices and $p,q(,r,...)$
over the $\m$-indices, and $t\in\bb{C}$ is an arbitrary complex number.
\end{defn}

%
\begin{defn}
Let $\theta(\g)=\g$ be a trivial involution of $\g$. Then the twisted Yangian $\Y(\g,\g)$ of the type II is a left coideal subalgebra of $\Y(\g)$ generated by the level-zero generators $j^{a}$ and the twisted level-two generators
\begin{align} \label{twistII}
\dwt{j}{}^{a} 
& = \wh{\wh{j}}{}^{a} + \alpha \, t \, \wh{j}^a+\frac{\alpha}{4}f^{a}_{\; bc}\big(\,\wh{j}^{b}j^{c}+j^{c}\wh{j}^{b}\big) \,,
\end{align}
having coproducts of the form
\begin{align} \label{twistIIcop}
\Delta(\dwt{j}{}^{a}) &= \dwt{j}{}^{a}\tc1+1\tc\dwt{j}{}^{a} +\alpha f^a_{\; bc}\,\wh{j}^{b}\tc j^{c} + \frac{\alpha^2}{4c_\g}f_{\; bc}^{a} \big(h^{\;\; cb}_{+\; lki}\, j^{l}j^{k}\tc j^{i} + h^{\;\; cb}_{-\; lki}\, j^{i}\tc j^{l}j^{k} \big) \,,
\end{align}
where $h^{\;\; cb}_{\pm\; lki} = f_{\; ld}^{c}f_{\; ke}^{b}f_{\sk i}^{de} \pm f_{\sk d}^{ce}(f_{\; ke}^{b}f_{\; li}^{d} + f_{\; le}^{b}f_{\; ki}^{d})$. Indices $a(,b,c,...)$ run over all indices of $\g$, and $t\in\bb{C}$ is an arbitrary complex number. 
\end{defn}

\begin{remark}
In the case when $c_{\g}=0$ (and $g_{ad}$ is degenerate) the twisted level-two generators can be alternatively defined by \cite{MReg}
\begin{align}
\dwt{j}{}^{cb} &= [\,\wh{j}^{c},\wh{j}^b] + \alpha \, t\,[{j}^{c},\wh{j}^b] +\frac{\alpha}{2}f^{b}_{\; de}f_{\sk g}^{ec}\big(\,\wh{j}^{g}j^{d}+j^{d}\wh{j}^{g}\big) \,.
\end{align}
\end{remark}

\begin{remark}

Let the coproduct of $\wh{j}$ be defined in terms of the Casimir operator $t$,
\begin{align}
\Delta(\wh{j}^a) = \wh{j}^a\tc1+1\tc\wh{j}^a + \frac{\alpha}{2}[t,j^a\tc1] \,,
\end{align}
where $t$ is regarded as an element of $\Uq^{\tc2}$. Then \eqref{twistI} can be written as
\begin{equation}
\wt{j}^{p} = \wh{j}^{p} + \alpha\, t\, j^p+\frac{\alpha}{4}[t^{\h},j^p],
\qquad \Delta(\wt{j}^{p}) = \wt{j}^{p}\tc1-1\tc\wt{j}^{p} + \alpha\, [t^\h,j^p\tc1]\,,
\end{equation}
where $t^\h$ is the restriction of $t$ to the subalgebra $\h$. The expressions \eqref{twistII} and \eqref{twistIIcop} can be written as
\begin{align}
\dwt{j}{}^{a} 
& = \wh{\wh{j}}{}^{a} + \alpha \, t \, \wh{j}^a+\frac{\alpha}{4}[t,\wh{j}^a] \,,\\
\Delta(\dwt{j}{}^{a}) &= \dwt{j}{}^{a}\tc1+1\tc\dwt{j}{}^{a} 
+\alpha [t,\wh{j}^a\tc1] + \frac{\alpha^2}{4c_\g}f^a_{\;bc}\,[[t,j^c\tc1],[t,j^b\tc1]] \,.
\end{align}

\end{remark}

The twisted Yangians with suitable $t\in\bb{C}$ are Yangian reflection algebras. The isomorphism with the FRT-realization of the twisted Yangians for the type I case for $\g=\sl(n),\;n>2$ was explicitly shown in \cite{GM}. In the following sections we will give explicit examples of the twisted Yangians of both types and show that they are Yangian reflection algebras. In section 3 we will construct generalized twisted Yangians for $\ytwo$, and in Section 5 for $\yone$. 
We will also show, however rather heuristically, that the quantum affine coideal subalgebras in the rational $q\to1$ limit specialize to the twisted Yangian of type I if $p(0)\neq0$, and to the type II if $p(0)=0$\,. It would be very interesting to see this specialization in terms of the approach presented in \cite{GTL,GM1}. However this would require Drinfeld second realization of the reflection algebras which to our knowledge is currently not known.


\vspace{0.1in}

\section{Reflection algebras for \texorpdfstring{$\sltwo$}{sl2}}

{\bf \noindent Algebra.}
The quantum affine Lie algebra $\sltwo$ in the Drinfeld-Jimbo realization is generated by the Chevalley generators $\xi^\pm_1$, the Cartan generator $k_1$, the affine Chevalley generators $\xi^\pm_0$ and the corresponding Cartan generator $k_0$. 
The extended (symmetric) Cartan matrix is given by
\begin{align}
(\wh{a}_{ij})_{0 \leq i,j \leq 1}=
\left(\begin{array}{ccc}
\;2&\!\!\!-2\\
\!\!\!-2&\;2\\
\end{array}\right) .
\end{align}
The corresponding root space $\Phi$ is generated by $\hat\pi = \{ \alpha_0 ,\, \alpha_1 \}$.
The commutation relations of the algebra are as follows,
\begin{align}
[k_i,k_j] &= 0\,, \qquad\quad k_i \,\xi_j^\pm = q^{\wh{a}_{ij}} \xi_j^\pm k_i \,, \qquad\quad [\xi^+_i,\xi^-_j] = \delta_{ij}\frac{k_i-k_i^{-1}}{q-q^{-1}}\,.
\end{align}
%

{\bf \noindent Representation.}
We define the fundamental evaluation representation $T_z$ of $\sltwo$ on a two-dimensional vector space $\cV$ with basis vectors $\{v_1,\, v_2\}$. Let $e_{j,k}$ be $2\times2$ matrices satisfying $(e_{j,k})_{j',k'} = \delta_{j,j'} \delta_{k,k'}$ or equivalently $e_{i,j}\,v_k = \delta_{j,k}\, v_i$ (i.e.\ for any operator $A$ its matrix elements $A_{ij}$ are defined by $A \, v_i = A_{ji} v_j$). Then the representation $T_z$ is given by 
\begin{align} \label{Tz:SL2q}
T_z(\xi^+_1) &= e_{1,2}\,,     &T_z(\xi^-_1) &= e_{2,1}\,,        &T_z(k_1) &= q\, e_{1,1}+q^{-1}e_{2,2}\,, \el 
T_z(\xi^+_0) &= z\, e_{2,1}\,, &T_z(\xi^-_0) &= z^{-1} e_{1,2}\,, &T_z(k_0) &= q^{-1}e_{1,1}+q\,e_{2,2}\,,
\end{align}
We choose the boundary vector space $\cW$ to be equivalent to $\cV$. The boundary representation $T_s$ is obtained from \eqref{Tz:SL2q} by replacing $z$ with $s$. Here $z$ (resp.\ $s$) is the bulk (resp.\ boundary) spectral parameter. 

The fundamental $R$-matrix $R_{ij}(z)\in End(\cV_i\tc \cV_j)$ satisfying the Yang-Baxter equation \eqref{YBE}
\begin{align}
R_{12}(z/w) R_{13}(z) R_{23}(w) = R_{23}(w) R_{13}(z) R_{12}(z/w) \,,\nonumber
\end{align}
is given by
\begin{align}
R(z) = \left(
\begin{array}{cccc} 
1 & 0     & 0      & 0 \\
0 & r     & 1-q\,r & 0 \\
0 & 1-r/q & r      & 0 \\
0 & 0     & 0      & 1
\end{array} \right) ,\qquad\text{where}\qquad r = \frac{z-1}{q \, z-1/q} \,. \label{Rsl2q}
\end{align}
%


\vspace{0.1in}

\subsection{Singlet boundary} \label{Sec:2.1}

Consider the reflection equation \eqref{RE} on the space $\cV \tc \cV \tc \bb{C} $ with the $R$-matrix defined by \eqref{Rsl2q} and the $K$-matrix being any $2\times2$ matrix satisfying
\begin{align} 
R_{12}(z/w) K_{13}(z) R_{12}(zw) K_{23}(w) = K_{23}(w) R_{12}(zw) K_{13}(z) R_{12}(z/w) \,.\nonumber
\end{align}
The general solution is \cite{GZ}
\begin{align}
K(z) = \left(
\begin{array}{cc} 
1      & a\, k' \\
b\, k' & k
\end{array} \right) ,\qquad\text{where}\qquad k = \frac{c\,z-1}{z(c-z)}\,,\qquad k' = \frac{1-z^2}{z(c-z)}\,, \label{GKsl2qs}
\end{align}
and $a,\,b,\,c\in\mathbb{C}$ are arbitrary complex numbers. 

We are interested in a solution compatible with the underlying Lie algebra. The minimal constraint is to require the reflection matrix to intertwine the Cartan generators,
\begin{align} \label{Int:sl2s-k}
(T_{1/z}\tc\epsilon)[\Delta(k_i)] \, K(z) - K(z) \, (T_{z}\tc\epsilon)[\Delta(k_i)] = 0 \,.
\end{align}
This constraint restricts the $K$-matrix \eqref{GKsl2qs} to be of a diagonal form ($a=b=0$)%
\footnote{This constraint may be alternatively obtained by requiring the unitarity property to hold, $K(z^{-1})K(z)=id$.}.
Next, it is easy to see that such a $K$-matrix does not satisfy the intertwining equation for any of the Chevalley generators, 
\begin{align} \label{Int:sl2s-xi}
(T_{1/z}\tc\epsilon)[\Delta(\xi^{\pm}_i)] \, K(z) - K(z) \, (T_{z}\tc\epsilon)[\Delta(\xi^{\pm}_i)] \neq 0 \,.
\end{align}
We call Cartan generators $k_i$ the \textit{preserved} generators, while the Chevalley generators $\xi^\pm_i$ are the \textit{broken} generators. This setup is consistent with the following quantum affine coideal subalgebra.

\begin{prop}
Let the involution $\Theta$ act on the root space $\Phi$ as
\begin{align} \label{Aut1}
\Theta(\alpha_0) = -\alpha_1 \,.
\end{align}
Then it defines a quantum affine coideal subalgebra $\cB\subset\cA=\sltwo$ generated by the Cartan element $k_0 k_1^{-1}$ and the twisted affine generators
\begin{align} \label{tw1}
B^{+}_0 &= \xi'^{+}_0 k_0 - d_{+} \, \xi^{-}_1 k_0 \,, \qquad\qquad
B^{-}_0 = \xi^{-}_0 k_0 - d_{-} \, \xi'^{+}_1 k_0 \,,
\end{align}
where $\xi'^{+}_i = k_i^{-1} \xi^{+}_i$ and $d_\pm \in \mathbb{C}$ are arbitrary complex numbers.
\end{prop}

\begin{proof}
The generators \eqref{tw1} satisfy the coideal property
\begin{align}
\Delta(B^{+}_0) &= \xi'^{+}_0 k_0 \tc 1 - d_{+} \, \xi^{-}_1 k_0 \tc k_0 k_1^{-1} + k_0  \tc B^{-}_0 \in \cA \tc \cB \,, \el
\Delta(B^{-}_0) &= \xi^{-}_0 k_0 \tc 1 - d_{-} \, \xi'^{+}_1 k_0 \tc k_0 k_1^{-1} + k_0  \tc B^{+}_0 \in \cA \tc \cB \,, \label{coideal1}
\end{align}
and the property is obvious for $k_0k_1^{-1}$.
\end{proof}


\begin{prop} \label{B1}
The quantum affine coideal subalgebra defined above with $d_+ \, q = d_- / q = c$, where $c \in \mathbb{C}$, is a reflection algebra for a \textit{singlet} boundary.   
\end{prop}

\begin{proof}
The representation of the generators of $\cB$ is given by
\begin{align}
T_{z}(k_0 k_1^{-1}) &= q^{-2} e_{1,1} + q^{2} e_{2,2} \,,  & T_{z}(B^{+}_0) &= q^{-2}( z - q \, d_+ ) e_{2,1} \,, & T_{z}(B^{-}_0) &= ( q \, z^{-1} - d_- ) e_{1,2} \,.
\end{align}
Let $K(z)$ be any $2\times 2$ matrix. The intertwining equation for $k_0 k_1^{-1}$ restricts $K(z)$ to be of a diagonal form, thus up to an overall scalar factor, $K(z)=e_{1,1}+k\,e_{2,2}$. Next, the intertwining equation for $B^\pm$ gives
\begin{align}
1 + q \,z \,d_+ (k-1) - z^2 k = 0 \,, \qquad d_- (k-1)+q \,(z^{-1}-z \, k) = 0 \,, 
\end{align}
having a unique solution $d_+ \, q = d_- / q = c$ and $k=\dfrac{c\, z-1}{z(c-z)}$, where $c\in\bb{C}$ is an arbitrary complex number. This coincides with \eqref{GKsl2qs} provided $a=b=0\,$. 

\end{proof}



\subsection{Vector boundary} \label{Sec:2.2}

Consider the reflection equation \eqref{RE} on the tensor space $\cV \tc \cV \tc \cW$ with the $R$-matrix defined by \eqref{Rsl2q}. Then there exists a solution of the reflection equation,
\begin{align}
K(z) = \left(
\begin{array}{cccc} 
1 & 0    & 0       & 0 \\
0 & 1-k/q  & k    & 0 \\
0 & k & 1-q\,k & 0 \\
0 & 0    & 0       & 1
\end{array} \right) ,\qquad\text{where}\qquad k = \frac{(q-q^{-1})(z^2-1)}{q^{-2} - c\,z + q^2 z^2} \,, \label{KVsl2q}
\end{align}
and $c \in \mathbb{C}$ is an arbitrary complex number. This $K$-matrix satisfies the intertwining equation \eqref{IntKV}
\begin{align}
(T_{1/z} \tc T_s)[\Delta(b)] \, K(z) - K(z) \, (T_{z} \tc T_s)[\Delta(b)] = 0  \quad\text{for all} \quad b \in \cU_q(\sl(2))\,. 
\end{align}
We call Cartan generators $k_i$ and Chevalley generators $\xi^\pm_1$ the \textit{preserved} generators, while the affine Chevalley generators $\xi^\pm_0$ are the \textit{broken} generators. Next, we identify the corresponding quantum affine coideal subalgebra consistent with the reflection matrix \eqref{KVsl2q}. 

\begin{prop} 
Let the involution $\Theta$ act on the root space $\Phi$ as
\begin{align} \label{Aut2}
\Theta(\alpha_0) = - \alpha_0 -2 \alpha_1 \,, \qquad\qquad \Theta(\alpha_1) = \alpha_1 \,.
\end{align}
Then it defines a quantum affine coideal subalgebra $\cB\subset\cA=\sltwo$ generated by the Cartan generator $k_1$, the Chevalley generators $\xi^\pm_1$, and the twisted affine generator
\begin{align} \label{tw2}
B^{-}_0 = \xi^{-}_0 k_0 - d_- \big(\big(\text{ad}_r\, \xi^+_1 \xi^+_1 \big)\xi'^{+}_0\big) k_0 \,,
\end{align}
where $\xi'^{+}_0 = k_0^{-1} \xi^{+}_0$ and $d_-\in\bb{C}^\times$ is an arbitrary non-zero complex number.

\end{prop}

\begin{proof}
The twisted affine generator \eqref{tw2} satisfies the coideal property
\begin{align}
\Delta(B^{-}_0) &= \xi^{-}_0 k_0 \tc 1 - d_- \big(\big(\text{ad}_r\, \xi^+_1 \xi^+_1 \big)\xi'^{+}_0\big) k_0 \tc k_1^{-2} + k_0 \tc B^{-}_0 \el
& \qquad + d_- q^2(q^2-q^{-2})\big( \xi^+_0 \tc (\text{ad}_r\,\xi^+_1)\xi'^+_1 - k_0\, (\text{ad}_r\,\xi^+_1)\xi'^+_0 \tc k_1^{-1} \xi'^+_1 \big) \in \cA \tc \cB \,.
\end{align}
The property is satisfied by the definition for the rest of the generators.

\end{proof}


\begin{remark}

This algebra may be alternatively generated by $k_1$, $\xi^\pm_1$, and the twisted affine generator
\begin{align} \label{tw2b}
B^{+}_0 = \xi'^{+}_0 k_0 - d_+ \big(\big(\text{ad}_r\, \xi^-_1 \xi^-_1 \big)\xi^{-}_0\big) k_0 \,,
\end{align} 
having coproduct
\begin{align}
\Delta(B^{+}_0) &= \xi'^{+}_0  k_0\tc 1 - d_+ \big(\big(\text{ad}_r\, \xi^-_1 \xi^-_1 \big)\xi^{-}_0\big) k_0 \tc k_1^{-2} + k_0  \tc B^{+}_0 \el
& \qquad - d_+ (q^2-q^{-2})\big( \xi^-_0 k_0 \tc (\text{ad}_r\,\xi^-_1)\xi^-_1 - q^{-2} \, k_0 \,(\text{ad}_r\,\xi^-_1)\xi^-_0 \tc k_1^{-1} \xi^-_1 \big) \in \cA \tc \cB \,,
\end{align}
and $d_+ = d_-^{\;-1}(q^{-1}+q)^{-2}$. The generators $B^\pm_0$ are related by
\begin{align}
B_0^- = - d_- \big[ (\adr\,\xi_1^+\xi_1^+) B_0^+ k_1 \big] k_1^{-1} \,.
\end{align}

\end{remark}

\begin{prop} \label{B2}
The quantum affine coideal subalgebra defined above with $q^2 d_+ = q^{-2} d_- = (q+q^{-1})^{-1}$ is a reflection algebra for a \textit{vector} boundary.   
\end{prop}

\begin{proof}
The representation $(T_z\tc T_s)$ of the coproducts of the Lie generators of $\cB$ is given by
\begin{align} \label{TAsl2v:Lie}
(T_{z}\tc T_s)[\Delta(k_1)]   &= q^2 e_{1,1} + e_{2,2} + e_{3,3} + q^{-2} e_{4,4} \,, \el
(T_{z}\tc T_s)[\Delta(\xi^+)] &= q\, e_{1,2} + e_{1,3} + e_{2,4} + q^{-1} e_{3,4} \,, \el
(T_{z}\tc T_s)[\Delta(\xi^-)] &= e_{2,1} + q^{-1} e_{3,1} + q\, e_{4,2} + e_{4,3} \,, 
\end{align}
and of the twisted affine generators by
\begin{align} \label{TAsl2v:TA}
(T_{z}\tc T_s)[\Delta(B^{+}_0)] &= \big(q^{\!-3}s+d_{+}(s^{\!-1}(q^{\!-2}+1)-z^{\!-1}(q^{\!-4}-1))\big)\, e_{2,1} + q^{\!-2}\big(z + d_{+} z^{\!-1}(q+q^{\!-1})\big)\, e_{3,1} \el 
& \quad + \big(q^{\!-2}z + d_{+}q^2 z^{\!-1}(q+q^{\!-1})\big) \, e_{4,2} + \big(q^{\!-1}s+d_{+}(s^{\!-1}(q^2+1)-z^{\!-1}(q^4-1))\big) \, e_{4,3} \,, \el
(T_{z}\tc T_s)[\Delta(B^{-}_0)] &= \big(s^{\!-1}+q^{\!-1}d_{-}(s(q^{\!-2}+1)- z(q^{\!-4}-1))\big) \, e_{1,2} + \big(q\,z^{\!-1} + d_{-} z (q^{\!-4}+q^{\!-2})\big) \, e_{2,4} \el 
& \quad + \big(q\,z^{\!-1} + d_{-} z (q^2+1)\big) \, e_{2,4} + \big(q^2 s^{\!-1} + q^{\!-1}d_{-}(s(q^2+1)-z(q^4-1)\big) \, e_{3,4} \,.
\end{align}
Let $K(z)$ be any $4\times4$ matrix. Then the intertwining equation for the Lie generators \eqref{TAsl2v:Lie} constrains $K(z)$ to the form given in \eqref{KVsl2q} up to an unknown function $k$ and an overall scalar factor. Next, the intertwining equation for the twisted affine generators \eqref{TAsl2v:TA} has a unique solution,
\begin{align}
q^2 d_{+} = q^{-2}d_{-} = (q+q^{-1})^{-1} \,, \qquad k=\frac{(q-q^{-1})(z^2-1)}{q^{-2}-(s^{-1}+s) z+q^2 z^2} \,,
\end{align}
which coincides with $k$ given in \eqref{KVsl2q} provided $c = s + s^{-1}$. 
\end{proof}


\begin{remark}
Reflection matrix \eqref{KVsl2q} satisfies the intertwining equation for all Cartan generators $k_i\in\cT$, thus the subalgebra $\cB\tc\cT\subset\cA$ is also a reflection algebra. The same is true for the reflection algebra of the singlet boundary.
\end{remark}


\begin{remark}
The coideal subalgebra defined in proposition \ref{B1} is also compatible with a vector boundary. The corresponding reflection matrix is 
\begin{equation} \label{GKsl2qsv}
K(z)=\left(
\begin{array}{cccc}
 1 & 0 & 0 & 0 \\
 0 & 1+s(q^2 z-c) k & q s(c-s) k & 0 \\
 0 & q(c s-1)k & k'+s(z^{-1}\!-q^2c)k & 0 \\
 0 & 0 & 0 & k' \\
\end{array}
\right),
\end{equation}
where
\begin{equation}
k=\frac{(q^2-1)(z^2-1)}{(c-z)(q^2 z-s)(q^2 s z-1)} \,,\qquad k' = \frac{c\,z-1}{z(c-z)}\,,
\end{equation}
and $c,s\in\bb{C}$ are arbitrary complex numbers. This vector boundary reflection matrix can be obtained using the boundary fusion procedure \cite{Sk,Mezincescu:1991ke}, which in this case is simply
$K_V(z) = P\,R(zs)\,(1\tc K_S(z))\,P\,R(z/s)$, where $K_V(z)$ is \eqref{GKsl2qsv}, $K_S(z)$ is \eqref{GKsl2qs} with $a=b=0$, and $P=R(0)$ with $R(z)$ given by \eqref{Rsl2q}.

Similarly, the coideal subalgebra defined in proposition \ref{B2} is compatible with the singlet boundary. However, the corresponding reflection matrix is trivial,
\begin{equation}
K=\left(
\begin{array}{cc}
 1 & 0 \\
 0 & 1 \\
\end{array}
\right).
\end{equation}
This reflection matrix can be obtained by solving the boundary intertwining equation for the Lie algebra generators only, and thus the twisted affine symmetries are redundant in this case. The reflection matrix \eqref{KVsl2q} can be obtained using an equivalent fusion procedure as above.
These properties will further reappear in the Yangian case, and for the $GL(1|1)$ algebra for both affine and Yangian cases. In these cases we will simply state that the corresponding reflection matrix is trivial and omit repeating the expression for the fusion procedure.

\end{remark}


\begin{remark} \label{remVl}
Let $T_{l,z}$ be a finite-dimensional irreducible representation of the algebra $\sltwo$. Let $l$ be an integral or half-integral non-negative number and $\Vl$ be a $(2l+1)$--dimensional complex vector space with a basis $\{v_m \,|\, m=-l,-l+1,\cdots,l\}$. For convenience we set $v_{-l-1}=v_{l+1}=0$. The operators $T_{l,z}(\xi^{\pm}_i),\;T_{l,z}(k_i)$ act on the space $\Vl$ by
\begin{align}
T_{l,z}(\xi^{\pm}_1)\,v_m &= \left([l\mp m]_q [l\pm m+1]_q\right)^{1/2} v_{m+1}\,, & T_{l,z}(k_1)\,v_m &= q^{2m} v_{m}\,,\el
T_{l,z}(\xi^{\pm}_0)\,v_m &= z^{\pm1} \left([l\mp m]_q [l\pm m+1]_q\right)^{1/2} v_{m+1}\,, & T_{l,z}(k_0)\,v_m &= q^{-2m} v_{m}\,, \label{lrepq}
\end{align}
Let the boundary vector space $\cW_l$ and the boundary representation $T_{l,s}$ be defined in the same way. Then all the constructions of the quantum affine coideal subalgebras presented above apply directly for any finite-dimensional representation $T_{l,z}$ and lead to a unique solution (for fixed $l$) of the reflection equation (for the singlet boundary this was explicitly shown in \cite{DN}).
\end{remark}

The coideal subalgebra given in proposition \ref{B1} is the augmented q-Onsager algebra $\overline{\mathcal{O}}_\hbar(\sl(2))$ of \cite{IT} (see also \cite{BB,BC}).
Both of subalgebras \ref{B1} and \ref{B2} by the construction are closely related to the orthogonal and symplectic twisted q-Yangians $\rm Y^{\rm tw}_q(\mathfrak{o}_2)$ and $\rm Y^{\rm tw}_q(\mathfrak{sp}_2)$ introduced in \cite{MRS}, however we do not know the exact isomorphism. 


\vspace{0.1in}

\section{Reflection algebras for \texorpdfstring{$\ytwo$}{ytwo}}

{\bf \noindent Algebra.}
The Yangian $\ytwo$ is generated by the level-zero Chevalley generators $E^\pm$, Cartan generator $H$, and the level-one Yangian generators $\wh{E}^\pm$ and the corresponding level-one Cartan generator $\wh{H}\,$. The commutation relations of the algebra are given by
\begin{align}
[H,E^\pm] &= \pm 2 E^\pm\,, & [E^+,E^-] &= H \,, & [H,\wh{E}^\pm] &= \pm 2 \wh{E}^\pm\,, &[E^\pm,\wh{E}^\mp] &= \pm\wh{H} \,, & [H,\wh{H}] &= 0\,.
\end{align}
The Hopf algebra structure is equipped with the following coproduct,
\begin{align} \label{cop2}
\Delta(H) &= H \tc1 + 1\tc H\,,
& \Delta(\wh{H}) &= \wh{H} \tc 1 + 1 \tc \wh{H} - \alpha\,( E^+ \tc E^- - E^- \tc E^+) \,,  \el
\Delta (E^\pm) &= E^\pm \tc 1 + 1 \tc E^\pm\,,  
& \Delta (\wh{E}^\pm) &= \wh{E}^\pm \tc 1 + 1 \tc \wh{E}^\pm \pm \frac{\alpha}{2}\,( E^{\pm}\tc H - H\tc E^{\pm}) \,.
\end{align}
%

{\bf \noindent Representation.}
The fundamental evaluation representation of $\ytwo$ on the two-dimensional vector space $\cV$ is defined by
\begin{align}
T_u(E^+) &= e_{1,2}\,,     & T_u(E^-) &= e_{2,1}\,,        & T_u(H) &= e_{1,1}-e_{2,2}\,,\el 
T_u(\wh{E}^+) &= u\, e_{1,2}\,, & T_u(\wh{E}^-) &= u \, e_{2,1}\,, &T_u(\wh{H}) &= u\,(e_{1,1}-e_{2,2}) \,.
\end{align}
We set $T_u(\alpha) = 1$. The boundary representation $T_s$ is obtained by replacing $u$ with $s$. Here $u$ (resp.\ $s$) is bulk (resp.\ boundary) spectral parameter.

The fundamental $R$-matrix $R_{ij}(u)\in End(\cV_i\tc \cV_j)$ satisfying the Yang-Baxter equation \eqref{YBEY} 
\begin{align} 
R_{12}(u-v) R_{13}(u) R_{23}(v) = R_{23}(v) R_{13}(u) R_{12}(u-v) \,, \nonumber
\end{align}
is given by
\begin{align}
R(u) = \left(
\begin{array}{cccc} 
1 & 0     & 0      & 0 \\
0 & r     & 1-r & 0 \\
0 & 1-r & r      & 0 \\
0 & 0     & 0      & 1
\end{array} \right) ,\qquad\text{where}\qquad r = \frac{u}{u-1} \,. \label{Rsl2}
\end{align}
%



\subsection{Singlet boundary} \label{Sec:3.1}

Consider the reflection equation \eqref{REY} on the space $\cV \tc \cV \tc \bb{C} $ with the $R$-matrix defined by \eqref{Rsl2} and the $K$-matrix being any $2\times2$ matrix satisfying
\begin{align}
R_{12}(u-v) K_{13}(u) R_{12}(u+v) K_{23}(v) = K_{23}(v) R_{12}(u+v) K_{13}(v) R_{12}(u-v) \,. \nonumber
\end{align}
The general solution is \cite{Cher}
\begin{align}
K(u) = \left(
\begin{array}{cc} 
1      & a\, k' \\
b\, k' & k
\end{array} \right) ,\qquad\text{where}\qquad k = \frac{c+u}{c-u}\,,\qquad k' = \frac{u}{c-u}\,, \label{GKsl2s}
\end{align}
and $a,\,b,\,c\in\mathbb{C}$ are arbitrary complex numbers. 

Once again we are interested in a solution compatible with the underlying Lie algebra and thus require the reflection matrix to intertwine the Cartan generator $H$,
\begin{align}
(T_u\tc\epsilon)[\Delta(H)] \, K(u) - K(u) \, (T_u\tc\epsilon)[\Delta(H)] = 0 \,.
\end{align}
This requirement restricts $K$-matrix \eqref{GKsl2qs} to be of the diagonal form ($a=b=0$). Next, it is easy to check that such a $K$-matrix does not satisfy the intertwining equation for any other generators of $\mathfrak{sl}(2)$ (and $\ytwo$\,).
Hence we call Cartan generator $H$ the \textit{preserved} generator, while the rest are the \textit{broken} generators. This setup allows to define the following involution, which is obvious and thus we do not give a proof of it, and the twisted Yangian (we will follow the same strategy in further sections).

\begin{prop}
Let the involution $\theta$ act on Lie algebra $\g=\mathfrak{sl}(2)$ as
\begin{align}  \label{Aut3}
\theta(H) = H \,, \qquad \theta(E^{\pm}) = -E^{\pm} \,,
\end{align}
defining a symmetric pair $(\g,\theta(\g))$, where $\theta(\g)$ is the $\theta$--fixed subalgebra of $\g$. Then $\theta$ can be extended to the $\bth$ involution of $\Yg$ such that
\begin{align} \label{t3}
\bth(\wh{H}) = -\wh{H} \,, \qquad \bth(\wh{E}^{\pm}) = \wh{E}^{\pm} \,, \qquad \bth(\alpha) = - \alpha \,.
\end{align}
\end{prop}


\begin{prop} \label{TY:sl2s}
The twisted Yangian $\Ygh$ of type I for $\g=\sl(2)$ and $\theta(\g)=H$ is the $\bth$--fixed coideal subalgebra of $\Yg$ generated by the Cartan generator $H$ and the twisted Yangian generators \cite{Ols}
\begin{align} \label{tw3}
\wt{E}^{\pm} = \wh{E}^\pm \pm \alpha\, t\, E^\pm \pm \frac{\alpha}{4}(E^{\pm}H + H\,E^{\pm}) \,.
\end{align}
Here $t\in\bb{C}$ is an arbitrary complex number. 
\end{prop}

\begin{proof}
The twisted generators \eqref{tw3} are in the positive eigenspace of the involution $\bth$ and satisfy the coideal property
\begin{align}
\Delta(\wt{E}^{\pm}) &= \wt{E}^{\pm} \tc 1 \pm \alpha\, E^{\pm} \tc H + 1 \tc \wt{E}^-_1 \in \Yg \tc \Ygh \,. 
\end{align}
The same properties for $H$ follows from the definition. 
\end{proof}

\begin{prop} 
The twisted Yangian $\Ygh$ defined above is a reflection algebra for the \textit{singlet} boundary.
\end{prop}

\begin{proof}
The representation $T_u$ of the generators of $\Ygh$ is given by
\begin{align}
T_{u}(H) &= e_{1,1} - e_{2,2} \,, \qquad T_{u}(\wt{E}^+) = (u+t) \, e_{1,2} \,, \qquad T_{u}(\wt{E}^-) = (u-t) \, e_{1,2} \,.
\end{align}
Let $K(u)$ be any $2\times2$ matrix. Then the intertwining equation for $H$ restricts $K(u)$ to be of the diagonal form, thus up to an overall scalar factor, $K(u)=e_{1,1}+k\,e_{2,2}$. Next, the intertwining equation for $\wt{E}^\pm$ has a unique solution $k=\dfrac{t+u}{t-u}$ which coincides with \eqref{GKsl2s} provided $c=t$ and $a=b=0$.
\end{proof}


\vspace{0.1in}

\subsection{Vector boundary} \label{Sec:3.2}

Consider the reflection equation \eqref{REY} on the tensor space $\cV \tc \cV \tc \cW$ with the $R$-matrix defined by \eqref{Rsl2}. Then there exists a solution of the reflection equation,
\begin{align}
K(u) = \left(
\begin{array}{cccc} 
1 & 0   & 0   & 0 \\
0 & 1-k & k   & 0 \\
0 & k   & 1-k & 0 \\
0 & 0   & 0   & 1
\end{array} \right) ,\qquad\text{where}\qquad k = \frac{2u}{c^2-(u-1)^2} \,, \label{KVsl2}
\end{align}
and $c \in \mathbb{C}$ is an arbitrary complex number.

This $K$-matrix satisfies the intertwining equation
\begin{align}
(T_{-u}\tc T_s)[\Delta(b)] \, K(u) - K(u) \, (T_{u}\tc T_s)[\Delta(b)] = 0 \quad\text{for all} \quad b \in \mathfrak{sl}(2)\,. 
\end{align}
Thus we call the level-zero $\sl(2)$ generators $E^\pm$ and $H$ the \textit{preserved} generators, while the level-one generators $\wh{E}^\pm$ and $\wh{H}$ are the \textit{broken} generators. This setup leads to the following involution and the twisted Yangian. 

\begin{prop} 
Let $\theta$ be a trivial involution of the Lie algebra $\g=\sl(2)$, 
\begin{align} \label{Aut4}
\theta(H) = H \,, \qquad \theta(E^{\pm}) = E^{\pm} \qquad \Longrightarrow \qquad \theta(\g) = \g \,.
\end{align}
Then it can be extended to a non-trivial involution $\bth$ of $\ytwo$ such that
\begin{align} \label{t4}
\bth(\wh{H}) = -\wh{H} \,, \qquad \bth(\wh{E}^{\pm}) = -\wh{E}^{\pm} \,, \qquad \bth(\alpha) = -\alpha\,.
\end{align}
\end{prop}


\begin{prop}  \label{TY:sl2v}
The twisted Yangian $\Ygg$ of type II for $\g=\sl(2)$ and $\theta(\g)=\g$ is the $\bth$--fixed coideal subalgebra of $\Yg$ generated by all level-zero generators and the level-two twisted Yangian generators
\begin{align} \label{tw4}
\wt{\wt{E}}{}^{\pm} &= \pm \frac{1}{2}\left( [\wh{H},\,\wh{E}^\pm] + \alpha\, \big(\pm t\, \wh{E}^\pm + H\,\wh{E}^\pm - E^\pm \wh{H} \big) \right) \qquad\text{and}\qquad \wt{\wt{H}} = [\,\wt{\wt{E}}{}^+,\,E^-] \,,
\end{align}
where $t\in\bb{C}$ is an arbitrary complex number.
\end{prop}

\begin{proof}
The twisted generators \eqref{tw4} are in the positive eigenspace of the involution $\bth$ and satisfy the coideal property
\begin{align}
\Delta\big(\,\wt{\wt{E}}{}^{\pm}\,\big) &= \wt{\wt{E}}{}^{\pm} \tc 1 + 1 \tc \wt{\wt{E}}{}^\pm \pm \alpha\, (\wh{E}^{\pm} \tc H - \wh{H} \tc E^{\pm}) + \frac{\alpha^2}{4} E^\pm \tc (H^2 \pm (t+2) H) \el
& \quad -\frac{\alpha^2}{4} \big( (4 E^\pm E^\mp +H^2 +(t+2) H)\tc E^+ + 4 E^\mp\tc E^\pm E^\pm + 2 H\tc E^\pm H \big)  \el 
& \in \Yg \tc \Ygg \,, \phantom{\frac{\alpha^2}{4}} \\
\Delta\big(\,\wt{\wt{H}}\,\big) &= \wt{\wt{H}} \tc 1 + 1 \tc \wt{\wt{H}} + 2 \alpha\, (\wh{E}^{-} \tc E^{+} - \wh{E}^{+} \tc E^{-} ) + \frac{\alpha^2}{4} H\tc (4E^- E^+ - H^2) \el
& \quad - \frac{\alpha^2}{4}\big((4 E^- E^+ + H^2)\tc H + 2\,t\, (E^+ \tc E^- -E^- \tc E^+ ) + 4 (E^+ \tc E^- H+E^- \tc E^+ H)\big)
\el 
& \in \Yg \tc \Ygg \,.
\end{align}
The same properties for the level-zero generators follow from the definition.
\end{proof}

\begin{prop}
The twisted Yangian $\Ygg$ defined above with $t=-2$ is a reflection algebra for a vector boundary.
\end{prop}

\begin{proof}
The representation $(T_u\tc T_s)$ of the coproducts of the Lie generators of $\Ygg$ is given by
\begin{align} \label{TYsl2:Lie}
(T_{u}\tc T_s)[\Delta(E^+)] &= e_{1,2} + e_{1,3} + e_{2,4} + e_{3,4} \,, & (T_{u}\tc T_s)[\Delta(H)]  &= 2\,( e_{1,1} + e_{4,4} )\,,\el
(T_{u}\tc T_s)[\Delta(E^-)] &= e_{2,1} + e_{3,1} + e_{4,2} + e_{4,3} \,, 
\end{align}
and of the twisted Yangian generators by
\begin{align} \label{TYsl2:TY}
(T_{u}\tc T_s)[\Delta(\dwt{E}{}^{+})] &= \alpha \, e_{1,2} + \beta \, e_{1,3} + \gamma \, e_{2,4} + \delta \, e_{3,4} \,, \el
(T_{u}\tc T_s)[\Delta(\dwt{E}{}^{-})] &= \delta \, e_{2,1} + \gamma \, e_{3,1} + \beta \, e_{4,2} + \alpha \, e_{4,3} \,, \el
(T_{u}\tc T_s)[\Delta(\dwt{H}\,)]\; &= \lambda\, e_{1,1} + \mu \, e_{2,2} - \mu\, e_{3,3} - \lambda\,e_{4,4} + \eta\,(e_{2,3}-e_{3,2})\,,
\end{align}
where
\begin{align}
\alpha &= \tfrac{1}{16}((4s+t+2)^2\!-(t+4)^2)-u-\tfrac{1}{2} \,, & \lambda &= \tfrac{1}{16}((4 u + t + 2)^2\! + (4 c + t + 2)^2\! - 2 (t + 2)^2)\! - \tfrac{1}{2} \,, \el
\beta &= \tfrac{1}{4}(2u+1) (2u+t+3) \,, & \mu &= \tfrac{1}{16}((4 u + t + 2)^2\! - (4 c + t + 2)^2) + 1 \,,\el
\gamma &= \tfrac{1}{4}(2u-1) (2u+t+1) \,, & \eta &= -(2u+t/2+1)  \,, \el
\delta &= \tfrac{1}{16}((4s+t+2)^2-t^2)+u-\tfrac{1}{2} \,.
\end{align}
Let $K(u)$ be any $4\times4$ matrix. Then the intertwining equation for the Lie generators \eqref{TYsl2:Lie} constrains $K(u)$ to the form given in \eqref{KVsl2} up to an unknown function $k$ and an overall scalar factor. Next, the intertwining equation for the twisted Yangian generators \eqref{TYsl2:TY} constrains $t=-2$ and has a unique solution $k = \dfrac{2u}{s^2-(u-1)^2}$ which coincides with \eqref{KVsl2} provided $c = s\,$.
\end{proof}


\begin{remark}

Following the same pattern as in the quantum affine case, the twisted Yangian given in proposition \ref{TY:sl2s} is also compatible with a vector boundary. The corresponding fundamental reflection matrix is 
\begin{equation}
K(z)=\left(
\begin{array}{cccc}
 1 & 0 & 0 & 0 \\
 0 & 1+(u-1-c)k & (c-s)k & 0 \\
 0 & (c+s)k & k'+(1-c-u)k & 0 \\
 0 & 0 & 0 & k' \\
\end{array}
\right),
\end{equation}
where
\begin{align}
k=\frac{2 u}{(c-u)(1+s-u)(u+s-1)} \,,\qquad k' = \frac{c+u}{c-u}\,,
\end{align}
and $c,s\in\bb{C}$ are arbitrary complex numbers.

The twisted Yangian given in proposition \ref{TY:sl2v} is compatible with the singlet boundary, and the corresponding reflection matrix is trivial.

\end{remark}


\begin{remark}
Let $T_{l,u}$ be a $(2l+1)$-dimensional representation of $\ytwo$. Let $\Vl$ be a vector space defined in remark \ref{remVl}. The generators of $\ytwo$ act on the space $\Vl$ by \cite{CP}
\begin{align}
T_{l,u}(E^{\pm})\,v_m &= \left((l\mp m) (l\pm m+1)\right)^{1/2} v_{m+1}\,, & T_{l,u}(H)\,v_m &= 2m\,v_{m}\,,\el
T_{l,u}(\wh{E}^{\pm}_0)\,v_m &= 2\,l\,u \left((l\mp m)(l\pm m+1)\right)^{1/2} v_{m+1}\,, & T_{l,u}(\wh{H})\,v_m &= 4m\,l\,u\, v_{m}\,, \label{lrep}
\end{align}
Let the boundary vector space $\cW_l$ and the boundary representation $T_{l,s}$ be defined in the same way. Then all the constructions of twisted Yangians presented above apply directly for any finite-dimensional representation $T_{l,u}$ and lead to unique solution (for fixed $l$) of the reflection equation.
\end{remark}

To finalize this section we want to note that the twisted Yangian of type I given in proposition \ref{TY:sl2s} is isomorphic to the orthogonal twisted Yangian $\mathcal{Y}^+(2)$ of \cite{Ols,MNO} and to $B(2,1)$ of \cite{MR}. The twisted Yangian of type II given in proposition \ref{TY:sl2v} is isomorphic to the symplectic twisted Yangian $\mathcal{Y}^-(2)$ of \cite{Ols,MNO} and to $B(2,0)$ of \cite{MR} (see also \cite{BC}).


\vspace{0.1in}

\subsection{Yangian limit} \label{Sec:3.3}

The algebra $\Ua$ does not contain any singular elements, and in the $q\to1$ limit specializes to $\cU(\g)$ via the composite map $\varphi$ such that 
\begin{align}
\Ua \underset{q\to1}{\longrightarrow} \cU(\cL(g)) \underset{z\to1}{\longrightarrow} \cU(g) \,.
\end{align}
Set $q={\rm e}^{\alpha \hbar}$ and $z={\rm e}^{-2\hbar u}$. Then the $q\to1$ limit is obtained by setting $\hbar\to0$, where $\hbar$ is an indeterminate deformation parameter that can be regarded as the Planck's constant when the Yangian is an auxiliary algebra and $\alpha$ is a formal parameter used to track the `level' of the Yangian generators and is usually set to unity \cite{MacKay:2004tc}. Consider an extended algebra 
\begin{align}
\wt\Ua = \Ua \tc_{\bb{C}[[\hbar]]} \bb{C}((\hbar)) \,.
\end{align}
Here $\bb{C}[[\hbar]]$ (resp.\ $\bb{C}((\hbar))$\,) denotes the formal power (resp.\ Laurent) series in $\hbar$. This algebra contains singular elements those that do not have a properly defined $q\to1$ limit. Let $\cA\subset\wt\Ua$ be the subalgebra generated $\Ua$ and $\hbar^{-1} {\rm ker}(\varphi)$. Then the Yangian $\Yg$ as an algebra is isomorphic to the quotient $\cA/\hbar\cA\cong\Yg$ \cite{Dr87}.

In such a way the Yangian $\ytwo$ can be obtained by taking the rational $q\to1$ limit of certain singular combinations of the generators of $\sltwo$. 
Here we shall be very concise and rely a lot on the evaluation map (for some heuristic arguments see appendix \ref{app:A}; for more thorough considerations see e.g.\ \cite{Tolstoy,GTL,GM1}).
The Yangian generators of $\ytwo$ are obtained by the following prescription,
\begin{align} \label{Ylimit}
\wh{E}^\pm = \pm \alpha\,\underset{q\to1}{\rm lim}\, \frac{\xi^\mp_0 - \xi^{\pm}_1}{q-q^{-1}} \,.
\end{align}
The Lie algebra generators are recovered by
\begin{align} \label{Llimit}
E^\pm = \underset{q\to1}{\rm lim}\, \xi^{\pm}_1 = \underset{q\to1}{\rm lim}\, \xi^{\mp}_0 
\qquad\text{and}\qquad 
H = \underset{q\to1}{\rm lim} \,h_1 = - \underset{q\to1}{\rm lim} \,h_0 \,.
\end{align}
We will next show that the quantum affine reflection algebras considered in sections \ref{Sec:2.1} and \ref{Sec:2.2} in the rational $q\to1$ limit specializes to the Yangian reflection algebras considered in sections \ref{Sec:3.1} and \ref{Sec:3.2} respectively.

\begin{prop}
The coideal subalgebra $\cB\subset\sltwo$ defined by the Proposition \ref{B1} in the rational $q\to1$ limit specializes to the twisted Yangian $\Y(\g,\h)$ of type I defined by the Proposition \ref{TY:sl2s}.
\end{prop}

\begin{proof}
Recall that $\cB$ is generated by the twisted affine generators \eqref{tw1} 
\begin{align}
B^{+}_0 = \xi'^+_0 k_0 - c / q\, \xi^-_1 k_0\,, \qquad 
B^{-}_0 = \xi^-_0 k_0 - c\,q\, \xi'^+_1 k_0 \,.
\end{align}
and the Cartan element $k_0 k_1^{-1}$. Note that $\xi'^+_0 k_0 = q^{-2} \xi^{+}_0$. We will be using the following expansion,
\begin{align} \label{q1exp}
k_i \to 1 + (q-1) h_i + \mathcal{O}(\alpha^2) \,, 
\end{align}
where $\mathcal{O}(\alpha^2)$ represent the higher order in $\alpha$ terms (here $\alpha\sim\hbar$).
Then by substituting $c\to q^{-2 t}$ and with the help of \eqref{Ylimit} and \eqref{Llimit} we find
\begin{align}
\underset{q\to1}{\rm lim}\, \frac{\alpha\,q^2 B_0^{+}}{q-q^{-1}} &= \underset{q\to1}{\rm lim} \left[ \alpha\,\frac{\xi^+_0 - \xi^-_1}{q-q^{-1}} + \alpha \, \frac{q-1}{q-q^{-1}}\big( 2 t \, \xi_1^- - \xi_1^- - \xi_1^- h_0 \big) + \mathcal{O}(\alpha^2) \right] \hspace{2.36cm} \el
& = -\wh{E}^{-} + \alpha\,t\,E^- + \frac{\alpha}{4}\,(E^{-}H + H\,E^-) = - \wt{E}^- \,,
\end{align}
and
\begin{align}
\underset{q\to1}{\rm lim}\, \frac{\alpha\,B_0^{-}}{q-q^{-1}} &= \underset{q\to1}{\rm lim} \left[ \alpha\,\frac{\xi^-_0 - \xi^+_1}{q-q^{-1}} + \alpha \, \frac{q-1}{q-q^{-1}}\big( \xi_0^- h_0 - (1-2 t)\xi_1^+ - \xi_1^+ h_0 + h_1\xi_1^+  \big) + \mathcal{O}(\alpha^2) \right] \el
& = \wh{E}^{+} + \alpha\,t\,E^{+} + \frac{\alpha}{4}\,(E^{+}H + H\,E^{+}) = \wt{E}^+ \,.
\end{align}
Finally,
\begin{align} \label{YL:sl2s_cartan}
\underset{q\to1}{\rm lim}\, \frac{1-k_0 k_1^{-1}}{q-q^{-1}} &= H \,.
\end{align}
These coincide with \eqref{tw3} as required.
\end{proof}

\begin{prop}
The quantum affine coideal subalgebra $\cB\subset\sltwo$ defined by the Proposition \ref{B2} in the rational $q\to1$ limit specializes to the twisted Yangian $\Y(\g,\g)$ of type II defined by the Proposition \ref{TY:sl2v}.
\end{prop}

\begin{proof}
Recall that $\cB$ is generated by all level-zero generators plus twisted affine generators \eqref{tw2} and  \eqref{tw2b},
\begin{equation}
B^{-}_0 = \xi^{-}_0 k_0 - d_- \big(\big(\text{ad}_r\, \xi^+_1 \xi^+_1 \big)\xi'^{+}_0\big) k_0 \,, \qquad
B^{+}_0 = \xi'^{+}_0 k_0 - d_+ \big(\big(\text{ad}_r\, \xi^-_1 \xi^-_1 \big)\xi^{-}_0\big) k_0 \,,
\end{equation} 
where $q^2 d_+ = q^{-2}d_- = (q+q^{-1})^{-1}$.
Next let us note that
\begin{equation}
[\wh{H},\wh{E}^\pm] = \pm \alpha\,\underset{q\to1}{\rm lim}\,\frac{[[\xi^\mp_0 - \xi^{\pm}_1,\xi^\mp_1],\xi^\mp_0 - \xi^{\pm}_1]}{2(q-q^{-1})^2} \,.
\end{equation}
This leads to the following expansion
\begin{align} \label{Xi0exp}
\xi^\mp_0 \to E^\pm \pm 2(q-1)\wh{E}^\pm \pm (q-1)^2([\wh{H}^\pm,\wh{E}^\pm] - \wh{E}^\pm) + \mathcal{O}(\alpha^3) \,.  
\end{align}
We will also need a higher order expansion of $k_i$, namely,
\begin{equation} \label{Kexp}
k_i \to 1 + (q-1) h_i + \tfrac{1}{2}(q-1)^2(h_i^2-h_i) + \mathcal{O}(\alpha^3) \,. 
\end{equation}

Consider the following rational combination,
\begin{equation}
\frac{q^{-2}\,B^{-}_0 - 2 \, \xi'^{+}_1}{(q-q^{-1})^2} 
= 
\frac{(1+q^2) (\xi^+_0-2 \xi^-_1) + q((1+q^2) \xi^-_1\xi^-_0\xi^-_1 - q^2 (\xi^-_1)^2\xi^-_0 - \xi^-_0(\xi^-_1)^2 )k_0 }{(q-q^{-1})^2(1+q^2)} \,.
\end{equation}
Then, by substituting \eqref{Xi0exp} and \eqref{Kexp} and with the help of \eqref{Ylimit} and \eqref{Llimit}, we find
\begin{equation}
\underset{q\to1}{\rm lim}\,\frac{q^{-2}\,B^{-}_0 - 2 \, \xi'^{+}_1}{(q-q^{-1})^2} 
= \tfrac{1}{2}\big( [\wh{H},\wh{E}^+] + \alpha\,(\wh{E}^+H - E^+ \wh{H})\big) -\tfrac{\alpha^2}{8}E^+(1+H)^2 = \dwt{E}{}^+ -\tfrac{\alpha^2}{4}E^+(1+H)^2
\,.
\end{equation}
Here we have used the constraint $t=-2$ for $\dwt{E}{}^+$ implicitly. In a similar way, for the twisted affine generator $B^+_0$, the rational combination
\begin{equation}
\frac{q^{2}\,B^{+}_0 - 2 \,\xi^{-}_1}{(q-q^{-1})^2} 
= 
\frac{q(1+q^2)(\xi^-_0k_0 - 2 \xi^+_1k_1^{-1}) + ((1+q^2) \xi^+_1 \xi^+_0 \xi^+_1  - q^2 \xi^+_0(\xi^+_1)^2 -(\xi^+_1)^2\xi^+_0)k_1^{-2} } {q^3(q-q^{-1})^2(1+q^2)} \,,
\end{equation}
leads to
\begin{equation}
\underset{q\to1}{\rm lim}\,\frac{q^{2}\,B^{+}_0 - 2 \,\xi^{-}_1}{(q-q^{-1})^2}  
= -\tfrac{1}{2}\big( [\wh{H},\wh{E}^-] + \alpha\,(\wh{E}^- H - E^- \wh{H})\big) -\tfrac{\alpha^2}{8}E^-(1-H)^2 = \dwt{E}{}^- -\tfrac{\alpha^2}{8}E^-(1-H)^2
\,.
\end{equation}
Finally, the Lie generators are recovered in the usual way as described above. 
\end{proof}


\smallskip

\section{Reflection algebras for \texorpdfstring{$\glone$}{glone}}

{\bf \noindent Algebra.}
The quantum affine Lie superalgebra $\glone$ in the Drinfeld-Jimbo realization is generated by the fermionic Chevalley generators $\xi^\pm_1$, the Cartan generators $k_1$, $k_2$ (here $k_2=q^{h_2}$, and $h_2$ is the non-supertraceless generator completing the superalgebra $\sl(1|1)$ to $\gl(1|1)\,$), and the affine fermionic Chevalley generators $\xi^\pm_0$ and the corresponding affine Cartan generator $k_0\,$. 
The extended (symmetric) Cartan matrix is given by
\begin{align} \label{GCMgl11}
(\wh{a}_{ij})_{0 \leq i,j \leq 2}=
\left(\begin{array}{cccc}
0&0&-2\\
0&0&2\\
-2&2&0
\end{array}\right) .
\end{align}
The corresponding root space has a basis of two fermionic roots, $\hat\pi = \{\alpha_0,\alpha_1\}.$ The commutation relations are as follows, for $0 \leq i,j \leq 2$ (Chevalley generators corresponding to the Cartan generator $k_2$ are absent):
\begin{align}
[k_i,k_j] = 0\,, \qquad\quad [k_i,\xi^\pm_j] = q^{\pm \wh{a}_{ij}} \xi^\pm_j\,,  \qquad\quad \{\xi^+_i,\xi^-_j\} = \delta_{ij}\frac{k_i-k_i^{-1}}{q-q^{-1}}\,.
\end{align}
Here $\{a,\,b\} = a\,b + b\,a\,$ denotes the anti-commutator. The graded right adjoint action is defined by
\begin{align}
\left(\adr \,\xi^{+}_i \right) a &= (-1)^{[\xi^{+}_i][a]}k_i^{-1} a\, \xi^{+}_i - k_i^{-1} \xi^{+}_i a \,, & \left(\adr \,\xi^{-}_i \right) a &= (-1)^{[\xi^{-}_i][a]}a\, \xi^{-}_i - \xi^{-}_i k_i \,a\, k_i^{-1} \,.
\end{align}
where $(-1)^{[\xi^{\pm}_i][a]}$ is the grading factor. Notice that the block $(\wh{a}_{ij})_{0 \leq i,j \leq 1}$ is trivial, and thus the right adjoint action for this block is equivalent to the regular graded commutator. This shall have important consequences for constructing the reflection algebra for a vector boundary. For this reason we shall also be in the need of a level-one Drinfeld generator $h^+_2$, which has coproduct defined by (see \cite{Zhang})
\begin{align} \label{dB11}
\Delta(\hp) & = \hp \tc 1 + 1 \tc \hp + 2 \, \xi^{+}_{0} k_1 \tc \xi^{+}_{1}.  
\end{align}

{\bf \noindent Representation.}
We define the fundamental evaluation representation $T_z$ of $\glone$ on a graded two-dimensional vector space $\cV$. Let $\cV=\{v_1,v_2\}$ and $\cV'=\{v'_1,v'_2\}$, then $v_1 v'_1 = v'_1 v_1$, $v_2 v'_2 = - v'_2 v_2$, and $v_2 v_2 = v'_2 v'_2 = 0$. Let $e_{j,k}$ be $2\times2$ matrices satisfying $(e_{j,k})_{j',k'} = \delta_{j,j'} \delta_{k,k'}$. 
Then the representation $T_z$ is defined by 
\begin{align} \label{Rep:gl11q}
T_z(\xi^+_1) &= e_{1,2}\,, & T_z(\xi^+_0) &= z\, e_{2,1}\,, & T_z(h_2) &= e_{1,1} - e_{2,2} \,,\el
T_z(\xi^-_1) &= e_{2,1}\,, & T_z(\xi^-_0) &= -z^{-1} e_{1,2}\,, & T_z(h^+_2) &= \frac{z\, q}{q^{-1}-q} (e_{1,1} - e_{2,2}) \,,\el
T_z(k_1)\, &= q\, e_{1,1}+q\, e_{2,2}\,, & T_z(k_0)\, &= q^{-1} e_{1,1}+ q^{-1}e_{2,2}\,, & T_z(k_2) &= q\,e_{1,1}+q^{-1}e_{2,2}\,.
\end{align}
We choose the boundary vector space $\cW$ to be equivalent to $\cV$. Then the boundary representation $T_s$ on $\cW$ is obtained from \eqref{Rep:gl11q} by replacing $z$ with $s$.

The fundamental $R$-matrix satisfying Yang-Baxter equation \eqref{YBE} is given by
\begin{align}
R(z) = \left(
\begin{array}{cccc} 
1 & 0     & 0      & 0 \\
0 & r     & 1-q\,r & 0 \\
0 & 1-r/q & r      & 0 \\
0 & 0     & 0      & -1+(q+1/q)\,r
\end{array} \right) ,\qquad\text{where}\qquad r = \frac{z-1}{q \, z-1/q} \,. \label{Rgl11q}
\end{align}
%


\vspace{0.1in}

\subsection{Singlet boundary} \label{Sec:4.1}

Consider the reflection equation \eqref{RE} on the space $\cV \tc \cV \tc \bb{C}$ with the $R$-matrix defined by \eqref{Rgl11q}. Then the general solution of the reflection equation is \cite{ZhangGL11}
\begin{align}
K(z) = \left(
\begin{array}{cc} 
1 & 0 \\
0 & k
\end{array} \right) ,\qquad\text{where}\qquad k = \frac{c\,z-1}{z(c-z)}\,, \label{GKgl11qs}
\end{align}
and $c\in\mathbb{C}$ is an arbitrary complex number. 

The general solution \eqref{GKgl11qs}, in contrast to \eqref{GKsl2qs}, is already of a diagonal form and thus intertwines all (level-zero) Cartan generators $k_i$, but does not satisfy the intertwining equation neither for any of the Chevalley generators nor for the level-one Cartan generators $h^\pm_2$. 
Hence we call Cartan generators $k_i$ the \textit{preserved} generators, while the generators $\xi^\pm_i$ and $h^\pm_2$ are the \textit{broken} generators. This setup is consistent with the following involution and the quantum affine coideal subalgebra.

\begin{prop}
Let the involution $\Theta$ act on the root space $\Phi$ as
\begin{align} \label{Aut5}
\Theta(\alpha_0) = -\alpha_1 \,.
\end{align}
Then it defines a quantum affine coideal subalgebra $\cB\subset\cA=\glone$
generated by the Cartan elements $k_2$, $k_0k_1^{-1}$, and the twisted affine generators
\begin{align} \label{tw5}
B^{+}_0 &= \xi'^{+}_0 k_0 - d_+ \, \xi^{-}_1 k_0 \,, \qquad\qquad
B^{-}_0 = \xi^{-}_0 k_0 - d_- \, \xi'^{+}_1 k_0 \,, 
\end{align}
where $\xi'^{+}_i = k_i^{-1}\xi^{+}_i$ and $d_\pm \in \mathbb{C}$ are arbitrary complex numbers.
\end{prop}

\begin{proof}
The coideal property is trivial for the Cartan elements, and for \eqref{tw5} follows directly from \eqref{coideal1}.
\end{proof}

\begin{prop}  \label{B5}
The quantum affine coideal subalgebra defined above with $d_+ \!= - d_-\! = q\,c$ where $c \in \mathbb{C}$ is an arbitrary complex number, is a reflection algebra for the \textit{singlet} boundary.   
\end{prop}

\begin{proof}
The representation $T_z$ of the generators of $\cB$ is given by 
\begin{align}
T_{z}(k_0 k_1^{-1}) &= q^{-2}(e_{1,1} + e_{2,2}) \,,  & T_{z}(B^{+}_0) &= (z-q^{-1}d_+)\,e_{2,1} \,, & T_{z}(B^{-}_0) &= q^{-2}(q\,z^{-1}-d_{-}) e_{1,2} \,.
\end{align}
and $T_z(k_2)$ was given in \eqref{Rep:gl11q}.
Let $K(z)$ be any $2\times 2$ matrix. Then the intertwining equation for $k_0 k_1^{-1}$ and $k_2$ restricts $K(z)$ to be of a diagonal form. This gives $K(z)= e_{1,1} + k\,e_{2,2}$ up to an overall scalar factor. Next, the intertwining equation for $B^\pm$ gives 
\begin{align}
d_+ (k-1) z - q(z^2k-1) = 0 \,, \qquad d_- (k-1) z+q (z^2 k-1) = 0 \,, 
\end{align}
having a unique solution $d_+ \!= - d_-\! = q\,c$ and $k=\dfrac{c\,z-1}{z(c-z)}$, where $c\in\bb{C}$ is any complex number. This coincides with \eqref{GKgl11qs}. 
\end{proof}


\smallskip

\subsection{Vector boundary} \label{Sec:4.2}

Consider the reflection equation \eqref{RE} on the tensor space $\cV \tc \cV \tc \cW$ with the $R$-matrix defined by \eqref{Rgl11q}. Then there exists a solution of the reflection equation
\begin{align}
K(z) = \left(
\begin{array}{cccc} 
1 & 0     & 0      & 0 \\
0 & 1-k/q & k      & 0 \\
0 & k     & 1-q\,k & 0 \\
0 & 0     & 0      & 1 - (q+q^{-1})\,k
\end{array} \right) ,\quad\text{where}\quad k = \frac{(q-q^{-1})(z^2-1)}{q^{-2} - c\,z + q^2 z^2} \,, \label{KVgl11q}
\end{align}
and $c \in \mathbb{C}$ is an arbitrary complex number.

This $K$-matrix satisfies the intertwining equation \eqref{IntKV} for all generators of the Lie superalgebra $\gl(1|1)$. Thus we call Cartan generators $k_i$ and Chevalley generators $\xi^\pm_1$ the \textit{preserved} generators, while the affine Chevalley generators $\xi^\pm_0$ and the level-one Cartan generators $h^{\pm}_2$ are the \textit{broken} generators. Next, we identify the corresponding quantum affine coideal subalgebra consistent with the reflection matrix \eqref{KVgl11q}. 

\begin{prop} 
Let the involution $\Theta$ act on the root space $\Phi$ as
\begin{align} \label{Aut6}
\Theta(\alpha_0) = -\alpha_0 -2\alpha_1 \,, \qquad\qquad \Theta(\alpha_1) = \alpha_1 \,.
\end{align}
Then it defines a quantum affine coideal subalgebra $\cB\subset\cA=\glone$ generated by the Cartan generators $k_i$, the Chevalley generators $\xi^{\pm}_1$, and the twisted affine generator
\begin{align} \label{tw6}
 B^{-}_0 &= \xi^{-}_0 - d_{-} [ h^{+}_2,\,\xi'^{+}_1 ] \,,
\end{align}
where $\xi'^{+}_i = k_i^{-1}\xi^{+}_i$ and $d_-\in\bb{C}^\times$ is an arbitrary complex number. 
\end{prop}

\begin{proof}
The twisted affine generator \eqref{tw6} satisfies the coideal property,
\begin{align}
\Delta(B^{-}_0) &= \xi^{-}_0 \tc k_0^{-1} - d_{-} [ h^{+}_2,\,\xi'^{+}_1 ] \tc k_1^{-1} + 1 \tc B^{-}_0 + 2 d_- \{\xi^+_0,\, \xi^+_1 \} \tc \xi'^+_1  \in \cA \tc \cB \,.
\end{align}
The property follows by definition for $k_i$ and $\xi_1^\pm$. 
\end{proof}

\begin{remark}
This algebra is not of the canonical form (compare \eqref{Bpm}, \eqref{tw2} and \eqref{tw6}). This is due to the all-zero entries in the block $0\leq i,j \leq 1$ of the extended Cartan matrix $\wh{a}_{ij}$ \eqref{GCMgl11}, and thus the $\adr$-action is equivalent to the usual graded commutator. In such a way the level-one Drinfeld generator $h_2^+$ is employed to ensure the coideal property. 
\end{remark}


\begin{prop}  \label{B6}
The quantum affine coideal subalgebra defined above with $d_-= (q^{-1}-q)/2$ is a reflection algebra for a \textit{vector} boundary.   
\end{prop}

\begin{proof}
The representation $(T_z\tc T_s)$ of the coproducts of the Lie generators of $\cB$ is given by
\begin{align} \label{CAgl11:Lie}
(T_{z}\tc T_s)[\Delta(\xi^+)] &= q\, e_{1,2} + e_{1,3} + e_{2,4} - q \, e_{3,4} \,, &
(T_{z}\tc T_s)[\Delta(k_0)]   &= q^{-2} (e_{1,1} + e_{2,2} + e_{3,3} + e_{4,4}) \,, \el
(T_{z}\tc T_s)[\Delta(\xi^-)] &= e_{2,1} + q^{-1} e_{3,1} + q^{-1} e_{4,2} - e_{4,3} \,, &
(T_{z}\tc T_s)[\Delta(k_1)]   &= q^2 \, (e_{1,1} + e_{2,2} + e_{3,3} + e_{4,4}) \,, \el
&& (T_{z}\tc T_s)[\Delta(k_2)]   &= q^2 e_{1,1} + e_{2,2} + e_{3,3} + q^{-2} e_{4,4} \,,
\end{align}
and of the twisted affine generator by
\begin{align} \label{CAgl11:TA}
(T_{z}\tc T_s)[\Delta(B^{-}_0)] &= \alpha \, e_{1,2} + \beta \, e_{1,3} + \alpha \, e_{2,4} - \alpha \, e_{3,4} \,,
\end{align}
where
\begin{align}
\alpha &= -q^{-1}(q\,s^{-1} + 2 d_{-}((q^{-2}-1)s - z)) \,,
& \beta &= -(q\,z^{-1} - 2 d_{-}(q^2-1)^{-1}z) \,.
\end{align}
Let $K(z)$ be any $4\times4$ matrix. Then the intertwining equation for the Lie generators \eqref{CAgl11:Lie} constrains $K(z)$ to the form given in \eqref{KVgl11q} up to an unknown function $k$ and an overall scalar factor. Next, the intertwining equation for the twisted affine generator \eqref{CAgl11:TA} has a unique solution
\begin{align}
d_{-} = (q^{-1}-q)/2 \,, \qquad k=\frac{(q-q^{-1})(z^2-1)}{q^{-2} - (s+s^{-1})z + q^2 z^2} \,,
\end{align}
which coincides with $k$ given in \eqref{KVgl11q} provided $c = s + s^{-1}$. 
\end{proof}


\begin{remark} \label{Ch2:remSB}

The coideal subalgebra given in proposition \ref{B5} does not lead to a unique reflection matrix for a vector boundary. The boundary intertwining equation in this case defines the reflection matrix up to an overall scalar factor and one unknown function,
\begin{equation}
K(z)=\left(
\begin{array}{cccc}
 1 & 0 & 0 & 0 \\
 0 & 1+(q^2 z-c)k' & q (c-s)k' & 0 \\
 0 & q (c-s^{-1})k' & 1+k & 0 \\
 0 & 0 & 0 & 1+k+(q^2 z^{-1}-c)k' \\
\end{array}
\right)
\end{equation}
where
\begin{equation}
k=\frac{q^2 \left(s + z(c-s) (c s-1)k' - s\, z^2\right)}{s\, z\, (q^2 z-c)} \,,
\end{equation}
and $c,s\in\bb{C}$ are arbitrary complex numbers. Here $k'=k'(z)$ is the unknown function that needs to be obtained by solving the reflection equation (having several solutions one of which is $k'=0$). In such a way this coideal subalgebra is not a reflection algebra for a vector boundary. It does not lead to a reflection matrix automatically satisfying the reflection equation.

The coideal subalgebra given in proposition \ref{B6} is compatible with the singlet boundary; the corresponding reflection matrix is trivial.
\end{remark}

Finally we note that the coideal subalgebra given in proposition \ref{B5} can be understood as the augmented q-Onsager algebra $\overline{\mathcal{O}}_\hbar(\gl(1|1))$.


\vspace{0.1in}

\section{Reflection algebras for \texorpdfstring{$\yone$}{yone}}

{\bf \noindent Algebra.}
The Yangian $\yone$ is generated by the level-zero Chevalley generators $E^\pm$, Cartan generators $H$ and the non-supertraceless generator $H_2$, and the level-one Yangian generators $\wh{E}^\pm$ and the corresponding level-one Cartan generators $\wh{H}$ and $\wh{H}_2\,$. The commutation relations of the algebra are
\begin{align}
& \{E^+,E^-\} = H \,, && \{E^\pm,\wh{E}^\mp\} = \wh{H} \,, && \!\!\!\! [H,E^\pm] = [H,\wh{E}^\pm] = [H,\wh{H}] = 0 \,,\el
& {[H_2,E^\pm]} = \pm 2 E^\pm \,, && [H_2,\wh{E}^\pm] = [\wh{H}_2,\,E^\pm] = \pm 2 \wh{E}^\pm \,.
\end{align}
The Hopf algebra structure is equipped with the following coproduct,
\begin{align} \label{cop4}
\Delta(H)\, &= H \tc 1 + 1 \tc H\,,
& \Delta (\wh{H})\, &= \wh{H} \tc 1 + 1 \tc \wh{H} \,, \el
\Delta (H_2) &= H_2 \tc 1 + 1 \tc H_2 \,,
& \Delta (\wh{H}_2) &= \wh{H}_2 \tc 1 + 1 \tc \wh{H}_2 - \alpha\,( E^+ \tc E^- + E^- \tc E^+) \,,  \el
\Delta (E^\pm) &= E^\pm \tc 1 + 1 \tc E^\pm\,,  
& \Delta (\wh{E}^\pm) &= \wh{E}^\pm \tc 1 + 1 \tc \wh{E}^\pm \mp \frac{\alpha}{2}\,( E^{\pm} \tc H - H \tc E^{\pm} ) \,.
\end{align}

{\bf \noindent Representation.}
The evaluation representation on the graded two-dimensional vector space $\cV$ is defined by
\begin{align} \label{Rep:gl11y}
T_u(E^+) &= e_{1,2}\,, & T_u(E^-) &= e_{2,1}\,,  & T_u(H) &= e_{1,1}+e_{2,2}\,, & T_u(H_2) &= e_{1,1}-e_{2,2}\,,\el 
T_u(\wh{E}^+) &= u\, e_{1,2}\,, & T_u(\wh{E}^-) &= u \, e_{2,1}\,, & T_u(\wh{H}) &= u\,(e_{1,1}+e_{2,2})\,, & T_u(\wh{H}_2) &= u\,(e_{1,1}-e_{2,2})\,,
\end{align}
and $T_u(\alpha)=1$. As previously, we choose the boundary vector space $\cW$ to be equivalent to $\cV$. The boundary representation $T_s$ on $\cW$ is obtained from \eqref{Rep:gl11y} by replacing $z$ with $s$.

The fundamental $R$-matrix satisfying Yang-Baxter equation \eqref{YBEY} is 
\begin{align}
R(u) = \left(
\begin{array}{cccc} 
1 & 0     & 0   & 0 \\
0 & r     & 1-r & 0 \\
0 & 1-r & r     & 0 \\
0 & 0     & 0   & -1+2r
\end{array} \right) ,\qquad\text{where}\qquad r = \frac{u}{u+1} \,. \label{Rgl11}
\end{align}
%


\vspace{0.1in}

\subsection{Singlet boundary} \label{Sec:5.1}

Consider the reflection equation \eqref{REY} on the tensor space $\cV \tc \cV \tc \bb{C}$ with the $R$-matrix defined by \eqref{Rgl11}. Then the general solution of the reflection equation is \cite{Arnaudon:2004sd}
\begin{align}
K(u) = \left(
\begin{array}{cc} 
1 & 0 \\
0 & k
\end{array} \right) ,\qquad\text{where}\qquad k = \frac{c+u}{c-u} \,, \label{GKgl11s}
\end{align}
and $c\in\mathbb{C}$ is an arbitrary complex number. 

This $K$-matrix intertwines the Cartan generators $H$ and $H_2$, but does not satisfy the intertwining equation for any other generators of $\gl(1|1)$ (and $\yone$\,). Hence we call Cartan generators $H$ and $H_2$ the \textit{preserved} generators, while the rest are the \textit{broken} generators. This setup allows to define the following twisted Yangian.

\begin{prop}
Let the involution $\theta$ act on the Lie algebra $\g=\gl(1|1)$ as
\begin{align}
\theta(H) = H \,, \qquad \theta(H_2) = H_2 \,, \qquad \theta(E^{\pm}) = -E^{\pm} \,,
\end{align}
defining a symmetric pair $(\g,\theta(\g))$. Then involution $\theta$ can be extended to the involution $\bth$ of $\Yg$ such that
\begin{align} \label{Aut7}
\bth(\wh{H}) = -\wh{H} \,, \qquad \bth(\wh{H}_2) = - \wh{H}_2 \,, \qquad \bth(\wh{E}^{\pm}) = \wh{E}^{\pm} \,, \qquad \bth(\alpha) = - \alpha \,.
\end{align}
\end{prop}


\begin{prop}  \label{TY:gl11s}
The twisted Yangian $\Ygh$ of type I for $\g=\gl(1|1)$ and $\theta(\g)=\{H,H_2\}$ is $\bth$--fixed coideal subalgebra of $\Yg$ generated by the Cartan generators $H$ and $H_2$, and the twisted Yangian generators
\begin{align} \label{tw7}
\wt{E}^{\pm} = \wh{E}^\pm \pm \alpha t \, E^\pm \mp \frac{\alpha}{2} H E^{\pm} \,, 
\end{align}
where $t\in\bb{C}$ is an arbitrary complex number.
\end{prop}

\begin{proof} 
The twisted generators \eqref{tw7} are in the positive eigenspace of the involution $\bth$ \eqref{Aut7} and satisfy the coideal property
\begin{align}
\Delta(\wt{E}^{\pm}) &= \wt{E}^{\pm} \tc 1  + 1 \tc \wt{E}^\pm \mp \alpha\, E^{\pm} \tc H \in \Yg \tc \Ygh \,. 
\end{align}
The same properties are obvious for $H$ and $H_2$.
\end{proof}

\begin{prop}
The twisted Yangian $\Ygh$ defined above is a reflection algebra for the singlet boundary.
\end{prop}

\begin{proof}
The representation $T_u$ of the twisted generators of $\Ygh$ is given by
\begin{align}
T_{u}(\wt{E}^+) &= (u+t-1/2) \, e_{1,2} \,, \qquad T_{u}(\wt{E}^-) = (u-t+1/2) \, e_{1,2} \,,
\end{align}
For $H$ and $H_2$ it was given in \eqref{Rep:gl11y}.
Let $K(u)$ be any $2\times2$ matrix. Then the intertwining equation for $H$ and $H_2$ restricts $K(u)$ to be of a diagonal form. This gives $K(u)= e_{1,1} + k\,e_{2,2}$ up to an overall scalar factor. Next, the intertwining equation for $\wt{E}^\pm$ has a unique solution $k=\dfrac{t+u-1/2}{t-u-1/2}$ which coincides with \eqref{GKgl11s} provided $t=c+1/2$ and $a=b=0$.
\end{proof}


\vspace{0.01in}

\subsection{Vector boundary} \label{Sec:5.2}

Consider the reflection equation \eqref{RE} on the tensor space $\cV \tc \cV \tc \cW$ with the $R$-matrix defined by \eqref{Rsl2}. Then there exists a solution of the reflection equation,
\begin{align}
K(u) = \left(
\begin{array}{cccc} 
1 & 0   & 0   & 0 \\
0 & 1-k & k   & 0 \\
0 & k   & 1-k & 0 \\
0 & 0   & 0   & 1-2k
\end{array} \right) ,\qquad\text{where}\qquad k = \frac{2 u}{(u+1)^2-c^2} \,, \label{KVgl11}
\end{align}
and $c \in \mathbb{C}$ is an arbitrary complex number.

This $K$-matrix satisfies the intertwining equation for all generators of $\gl(1|1)$.
We call the level-zero generators $E^\pm$, $H$ and $H_2$ the \textit{preserved} generators, while the level-one generators $\wh{E}^\pm$, $\wh{H}$ and $\wh{H}_2$ are the \textit{broken} generators. This setup leads to the following twisted Yangian. 

\begin{prop}
Let $\theta$ be a trivial involution of the Lie algebra $\g=\gl(1|1)$, 
\begin{align} \label{Aut8}
\theta(H) = H \,, \qquad \theta(H_2) = H_2 \,, \qquad \theta(E^{\pm}) = E^{\pm} \qquad \Longrightarrow \qquad \theta(\g) = \g \,.
\end{align}
Then it can be extended to a non-trivial involution $\tth$ of $\yone$ such that
\begin{align} \label{ty8}
\bth(\wh{H}) = -\wh{H} \,, \qquad \bth(\wh{H}_2) = -\wh{H}_2 \,, \qquad \bth(\wh{E}^{\pm}) = -\wh{E}^{\pm} \,, \qquad \bth(\alpha) = -\alpha\,.
\end{align}
\end{prop}


\begin{prop} \label{TY:gl11v}
The twisted Yangian $\Ygg$ of type II for $\g=\gl(1|1)$ and $\theta(\g)=\g$ is $\bth$--fixed coideal subalgebra of $\Yg$ generated by the level-zero generators and the level-two twisted Yangian generators
\begin{align} \label{tw8}
\dwt{E}{}^{\pm} &= \pm \frac{1}{2}\left([\wh{H}_2,\,\wh{E}^\pm] + \alpha\, (\pm t\,\wh{E}^\pm + E^\pm \wh{H} - \wh{E}^\pm H) \right) 
\qquad\text{and}\qquad 
\dwt{H} = \{E^\pm,\,\dwt{E}{}^\mp\} \,,  
\end{align}
where $t\in\bb{C}$ is any complex number.
\end{prop}

\begin{proof}
The twisted generators \eqref{tw8} are in the positive eigenspace of the involution $\bth$ \eqref{ty8} and satisfy the coideal property
\begin{align}
\Delta(\dwt{E}{}^{\pm}) &= \dwt{E}{}^{\pm} \tc 1 + 1 \tc \dwt{E}{}^\pm \mp \alpha\, (\wh{E}^{\pm} \tc H - \wh{H} \tc E^{\pm}) \el 
& \quad + \frac{\alpha^2}{4}(E^\pm\tc( H^2 - t\,H ) - H^2\tc E^\pm - 2 H\tc ( E^\pm H - 2\,t\,E^\pm)  \el
 & \in \Yg \tc \Ygg \,, \phantom{\frac{\alpha^2}{4}} \\
\Delta(\dwt{H}) &= \dwt{H} \tc 1 + 1 \tc \dwt{H} - \frac{\alpha^2}{4}(H^2\tc H + H\tc H^2) \el
 & \in \Yg \tc \Ygg \,, \phantom{\frac{\alpha^2}{4}}
\end{align}
and for the level-zero generators this property follow identically. 
\end{proof}

\begin{prop} 
The twisted Yangian $\Ygg$ defined above with $t=0$ is a reflection algebra for a vector boundary.
\end{prop}

\begin{proof}
The representation $(T_u\tc T_s)$ of the coproducts of the Lie generators of $\Ygg$ is given by
\begin{align} \label{TYgl11:Lie}
(T_{u}\tc T_s)[\Delta(E^+)] &= e_{1,2} + e_{1,3} + e_{2,4} - e_{3,4} \,, 
& (T_{u}\tc T_s)[\Delta(H)]\;\, &= 2\, \textstyle{\sum_{i=1}^4} \, e_{i,i} \,, \el
(T_{u}\tc T_s)[\Delta(E^-)] &= e_{2,1} + e_{3,1} + e_{4,2} - e_{4,3} \,,
& (T_{u}\tc T_s)[\Delta(H_2)] &= 2\,( e_{1,1} - e_{4,4} )\,,
\end{align}
and of the twisted Yangian generators by
\begin{align} \label{TYgl11:L2}
(T_{z}\tc T_s)[\Delta(\dwt{E}{}^+)] &= \alpha \, e_{1,2} + \beta \, e_{1,3} + \beta \, e_{2,4} - \alpha \, e_{3,4} \,,
& (T_{z}\tc T_s)[\Delta(\dwt{H})] &= \eta \, \textstyle{\sum_{i=1}^4} \, e_{i,i}  \,, \el
(T_{z}\tc T_s)[\Delta(\dwt{E}{}^-)] &= \gamma \, e_{2,1} + \delta \, e_{3,1} + \gamma \, e_{4,2} - \gamma \, e_{4,3} \,, 
\end{align}
where
\begin{align}
\alpha &= 2 s (2 s + t) + 4u + t-3 \,, & \beta &= (2u-1)(2u+t-1), & \eta &= \tfrac{1}{16}((4s + t)^2\! + (4u + t)^2\! - 2t^2\! - 8),\el
\gamma &= 2s(2s+t) -4u -t-3 \,, & \delta &= (2u+1)(2u+t+1).
\end{align}
Let $K(u)$ be any $4\times4$ matrix. The intertwining equation for the Lie generators \eqref{TYgl11:Lie} constrains $K(u)$ to the form given in \eqref{KVgl11} up to an unknown function $k$ and an overall scalar factor. Next, the intertwining equation for the twisted Yangian generators \eqref{TYgl11:L2} constrains $t=0$ and has a unique solution $k = \dfrac{2u}{(u+1)^2-s^2}$ which coincides with \eqref{KVgl11} provided $c = s\,$.
\end{proof}


\begin{remark}

The twisted Yangian given in proposition \ref{TY:gl11s} does not lead to a unique reflection matrix for a vector boundary. The boundary intertwining equation in this case defines the reflection matrix up to an overall scalar factor and one unknown function,
\begin{equation}
K(z)=\left(
\begin{array}{cccc}
 1 & 0 & 0 & 0 \\
 0 & 1+(1-c+u)k' & (c-s)k' & 0 \\
 0 & (c+s)k' & 1+k & 0 \\
 0 & 0 & 0 & 1+k+(1-c-u)k' \\
\end{array}
\right)
\end{equation}
where
\begin{equation}
k=\frac{(c^2-s^2)k'-2 u}{1-c+u} \,, \qquad c=t-1/2\,,
\end{equation}
and $t,s\in\bb{C}$ are arbitrary complex numbers. Here $k'=k'(u)$ is an unknown function that needs further to be obtained by solving the reflection equation (giving i.e.\ $k'=0$ or $k'=(2 u)/(s^2+c (1+u)^2-u (1+u)^2)$). In such a way this twisted Yangian is not a reflection algebra for a vector boundary. It does not lead to a reflection matrix automatically satisfying the reflection equation. This is in agreement with the remark \ref{Ch2:remSB}.

The twisted Yangian given in proposition \ref{TY:gl11v} is compatible with the singlet boundary; the corresponding reflection matrix is trivial.

\end{remark}


\smallskip

\subsection{Yangian limit} \label{Sec:5.3}

Lie algebras $\sl(2)$ and $\gl(1|1)$ are very similar, thus the Yangian limit of coideal subalgebras of $\glone$ is obtained in a very similar way as it was done in section \ref{Sec:3.3}. The Yangian $\yone$ can be obtained by taking the rational $q\to1$ limit of singular combinations of the generators of $\glone$,
\begin{align}
\wh{E}^\pm = \alpha\, \underset{q\to1}{\rm lim}\, \frac{\xi^\mp_0 \pm \xi^{\pm}_1}{q-q^{-1}} \,, 
\end{align}
The Lie algebra generators are given by
\begin{align}
E^\pm = \underset{q\to1}{\rm lim}\, \xi^{\pm}_1 = \mp \underset{q\to1}{\rm lim}\, \xi^{\mp}_0 \qquad\text{and}\qquad H = \underset{q\to1}{\rm lim} \,h_1 = - \underset{q\to1}{\rm lim} \,h_0 \,.
\end{align}
Next will show that the quantum affine reflection algebra considered in section \ref{Sec:4.1} in the rational $q\to1$ limit specializes to the Yangian reflection algebra considered in section \ref{Sec:5.1}. We will not attempt to show the specialization of the reflection algebra presented in section \ref{Sec:4.2}, as this requires considerations of the Drinfeld second realization of $\glone$ and goes beyond the scope of the present work.


\begin{prop}
The quantum affine coideal subalgebra $\cB\subset\glone$ defined by the Proposition \ref{B5} in the rational $q\to1$ limit specializes to the twisted Yangian $\Y(\g,\h)$ of type I defined by the Proposition \ref{TY:gl11s}.
\end{prop}

\begin{proof}
Recall that $\cB$ is generated by the twisted affine generators \eqref{tw5} 
\begin{align}
B^{+}_0 = \xi'^+_0 k_0 - c\,q\, \xi^-_1 k_0\,, \qquad 
B^{-}_0 = \xi^-_0 k_0 + c\,q\, \xi'^+_1 k_0 \,.
\end{align}
and Cartan element $k_0 k_1^{-1}$. Firstly note that $\xi'^+_0 k_0 = \xi^{+}_0$. Then by substituting $c\to q^{2(t-1/2)}$ and using the expansion \eqref{q1exp} we find
\begin{align}
\underset{q\to1}{\rm lim}\, \frac{\alpha\, B_0^{+}}{q-q^{-1}} &= \underset{q\to1}{\rm lim} \left[ \alpha\,\frac{\xi^+_0 - \xi^-_1}{q-q^{-1}} - \alpha \, \frac{q-1}{q-q^{-1}}\big( 2\,t\,\xi_1^{-} + \xi_1^{-} h_0 \big) + \mathcal{O}(\alpha^2) \right] \hspace{2.27cm} \el
& = \wh{E}^{-} - \alpha\,t\,E^- + \frac{\alpha}{2} H E^{-} = \wt{E}^- \,,
\end{align}
and
\begin{align}
\underset{q\to1}{\rm lim}\, \frac{\alpha\, B_0^{-}}{q-q^{-1}} &= \underset{q\to1}{\rm lim} \left[ \alpha\,\frac{\xi^-_0 + \xi^+_1}{q-q^{-1}} + \alpha \, \frac{q-1}{q-q^{-1}}\big(2\,t\,\xi_1^{+} + \xi_0^{-} h_0 + \xi^+_1 (h_0-h_1) \big) + \mathcal{O}(\alpha^2) \right] \el
& = \wh{E}^{+} + \alpha\,t\,E^{+} - \frac{\alpha}{2} H E^{+} = \wt{E}^+ \,.
\end{align}
Finally, $H_2 = \underset{q\to1}{\rm lim}\, ({k^2_2-1})/({q-q^{-1}}) $, and the Cartan generator $H$ is obtained in an equivalent way as in \eqref{YL:sl2s_cartan}. These coincide with \eqref{tw7} as required.
\end{proof}


\smallskip


\appendix

\section{The Yangian limit of \texorpdfstring{$\sltwo$}{sltwo}} \label{app:A}

Here we give a heuristic derivation of the Yangian limit of $\sltwo$. We will recover the level-one generators of $\ytwo$ by calculating the rational $q\to1$ limit of certain singular combinations of the generators of $\sltwo$ that otherwise do not have a propierly defined $q\to1$ limit. 

Set $\g=\sl(2)$. Recall that $\hat\sl(2)\cong\cL(\g)$. Let the Chevalley basis of $\cU(\g[u])$ be given by
\begin{align}
E^+ &= 1\tc e \,,\hspace{-1cm} & E^- &= 1\tc f \,,\hspace{-1cm} & H &= 1\tc h \,, \el
\hat{E}^+ &= u\tc f \,,\hspace{-1cm} & \hat{E}^- &= -u\tc e \,,\hspace{-1cm} & \hat{H} &= u\tc h \,.
\end{align}
Choose $e_q$, $f_q$ and $k_q=q^h$ be the basis of $\cU_q(\g)$. Then the basis of $\cU_q(\cL(g))$ is given by
\begin{align}
\hspace{0.7cm}\xi^+_1 &= 1\tc e_q \,,& \xi^-_1 &= 1\tc f_q \,,& h_1 &= 1\tc h\,,& k_1 &= q^{h_1} \,,\hspace{1.2cm} \el
\hspace{0.7cm}\xi^+_0 &= z\tc f \,,& \xi^-_0 &= z^{-1}\tc e \,,& h_0 &= -h_1 \,,& k_0 &= q^{h_0} \,.\hspace{1.2cm}
\end{align}

Set $q={\rm e}^{\alpha \hbar}$ and $z={\rm e}^{-2\hbar u}$. Then the $q\to1$ limit is obtained by setting $\hbar\to0$. The expansion in series of the Cartan generators $k_i$ at the point $\hbar=0$ is given by
\begin{equation}
k_i \to 1 + (q-1) h_i + \mathcal{O}(\alpha^2) \,, 
\end{equation}
where $\mathcal{O}(\alpha^2)$ denotes higher order in $\alpha$ terms ($\alpha\sim\hbar$). Note  that the (full) $q\to1$ limit gives
\begin{equation} 
E^+ = \underset{q\to1}{\rm lim}\,\xi^+_1 = \underset{q\to1}{\rm lim}\,\xi^-_0 \,, \qquad\quad 
E^- = \underset{q\to1}{\rm lim}\,\xi^-_1 = \underset{q\to1}{\rm lim}\,\xi^+_0 \,.
\end{equation}
Choose
\begin{equation}
\xi^+_{\alpha} = \frac{\alpha}{q-q^{-1}}(\xi^-_0 - \xi^+_1) \,, \qquad\quad \xi^-_{\alpha} = -\frac{\alpha}{q-q^{-1}}(\xi^+_0 - \xi^-_1) \,.
\end{equation}
Then
\begin{align}
\underset{q\to1}{\rm lim}\,\xi^+_{\alpha} &= \underset{q\to1}{\rm lim}\,\alpha\,\frac{z^{-1}-1}{q-q^{-1}}\tc e_q = u\tc e = \hat{E}^+_\alpha \,, \el
\underset{q\to1}{\rm lim}\,\xi^-_{\alpha} &= \underset{q\to1}{\rm lim}\,\alpha\,\frac{1-z}{q-q^{-1}}\tc f_q = -u \tc f = \hat{E}^-_\alpha \,,
\end{align}
and
\begin{align}
\underset{q\to1}{\rm lim}\,\Delta(\xi^+_{\alpha}) &= 
\underset{q\to1}{\rm lim}\,\frac{\alpha}{q-q^{-1}} (\xi^-_0 \tc k^{-1}_0 + 1\tc\xi^-_0 - \xi^+_1 \tc 1 - k_1\tc\xi^+_1) \el
 &= \underset{q\to1}{\rm lim}\, \left[ \alpha\,\frac{\xi^-_0 - \xi^+_1}{q-q^{-1}} \tc 1 + 1\tc \alpha\,\frac{\xi^-_0 - \xi^+_1}{q-q^{-1}} - \alpha\,\frac{q-1}{q-q^{-1}}( \xi^-_0 \tc h_0 + h_1\tc\xi^+_1) +\mathcal{O}(\alpha^2) \right] \el
 &= \hat{E}^{+}_\alpha\tc1 + 1\tc\hat{E}^{+}_\alpha + \frac{\alpha}{2} (E^+\tc H - H\tc E^+) = \Delta(\hat{E}^{+}_\alpha) \,, 
\end{align}
\begin{align}
\underset{q\to1}{\rm lim}\,\Delta(\xi^-_{\alpha}) &= \underset{q\to1}{\rm lim}\,\frac{\alpha}{q-q^{-1}} (- \xi^+_0 \tc 1 - k_0\tc\xi^+_0 + \xi^-_1 \tc k^{-1}_1 + 1\tc\xi^-_1) \el 
 &= \underset{q\to1}{\rm lim}\, \left[ \alpha\,\frac{\xi^-_1-\xi^+_0}{q-q^{-1}} \tc 1 + 1\tc \alpha\,\frac{\xi^-_1-\xi^+_0}{q-q^{-1}} - \alpha\,\frac{q-1}{q-q^{-1}}( \xi^-_1 \tc h_1 + h_0\tc\xi^-_0) +\mathcal{O}(\alpha^2) \right] \el
 &= \hat{E}^{-}_\alpha\tc1 + 1\tc\hat{E}^{-}_\alpha - \frac{\alpha}{2} (E^-\tc H - H\tc E^+) = \Delta(\hat{E}^{-}_\alpha) \,,
\end{align}
which coincide with \eqref{cop2}.

\smallskip


\bibliographystyle{amsplain}

\end{document}